\documentclass[journal]{IEEEtran}
\usepackage{amsmath}
\usepackage{amssymb}
\usepackage{amsthm}
\usepackage{mathrsfs}
\usepackage{mathtools} %in order to use \bigtimes
\usepackage{verbatim}

\newcommand\pgen{p_{\mathrm{gen}}}
\newcommand\pswap{p_{\mathrm{swap}}}

\newcommand\ttrunc{t_{\mathrm{trunc}}}

\usepackage{amsmath}  % in order to use 'align'
\usepackage[ruled]{algorithm2e}
\usepackage[shortlabels]{enumitem}  % to change the item label in an enumeration
\usepackage{datetime}  % may be removed: only to print the current time

\usepackage[T1]{fontenc}

\newtheorem{theorem}{Theorem}[section]

\newtheorem{lemma}[theorem]{Lemma}

\newtheorem{prop}{Proposition}[section]

\usepackage{thm-restate}

\usepackage{physics}
\usepackage{hyperref}

\usepackage{graphicx}
\usepackage{subcaption}
\usepackage{amsmath}
\usepackage{amssymb}
\usepackage{multirow}
\usepackage{cite} % for compressing numeric citation, i.e. [3]-[6] instead of [3],[4],[5],[6]

\usepackage{todonotes}
\newcommand{\unit}{1\!\!1}

\newcommand{\sumforget}[2]{{\widehat{\sum}_{#1}^{#2}}}

\newcommand{\mean}[1]{E\left[#1\right]}
\newcommand{\meanpartial}[2]{E\left[#1, #2\right]}
\def\stoch{\leq_{\textnormal{st}}}

\def\Tupper{T^{\textnormal{upper}}}

\def\Mupper{M^{\textnormal{upper}}}
\def\werner{w}
\newcommand{\Conv}{%
	  \mathop{\scalebox{3.0}{\raisebox{-0.2ex}{$*$}}
	    }
	    }

\let\oldnl\nl% Store \nl in \oldnl
\newcommand{\nonl}{\renewcommand{\nl}{\let\nl\oldnl}}% Remove line number for one line
\def\genswaponly{\mbox{\sc swap-only}}

\def\ddistillation{\mbox{$d$-{\sc dist-swap}}}

\definecolor{ballblue}{rgb}{0.13, 0.67, 0.8}
\definecolor{brandeisblue}{rgb}{0.0, 0.44, 1.0}

\begin{document}
\title{Efficient computation of the waiting time and fidelity in quantum repeater chains
}
\author{
	Sebastiaan Brand${}^{*\dagger}$,
	Tim Coopmans${}^{*\ddagger}$,
	David Elkouss${}^{\ddagger}$\thanks{${}^*$These authors contributed equally.}
\thanks{
${}^\dagger$
S. Brand was with Leiden Institute of Advanced Computer Science, Leiden University, Leiden, The Netherlands.
	} 
\thanks{${}^\ddagger$
T. Coopmans and D. Elkouss are with QuTech, Delft University of Technology, Delft, The Netherlands.
	}
\thanks{
Some of the ideas in this work were presented earlier in the master thesis of Sebastiaan Brand (6 June 2019), which can be found online on \mbox{\url{https://theses.liacs.nl/pdf/2018-2019-BrandSebastiaan.pdf}}
}
	}

\maketitle

\begin{abstract}
Quantum communication enables a host of applications that cannot be achieved by classical communication means, with provably secure communication as one of the prime examples.
The distance that quantum communication schemes can cover via direct communication is fundamentally limited by losses on the communication channel.
By means of quantum repeaters, the reach of these schemes can be extended and chains of quantum repeaters could in principle cover arbitrarily long distances.
In this work, we provide two efficient algorithms for determining the generation time and fidelity of the first generated entangled pair between the end nodes of a quantum repeater chain.
The runtime of the algorithms increases polynomially with the number of segments of the chain, which improves upon the exponential runtime of existing algorithms.
Our first algorithm is probabilistic and can analyze refined versions of repeater chain protocols which include intermediate entanglement distillation.
Our second algorithm computes the waiting time distribution up to a pre-specified truncation time, has faster runtime than the first one and is moreover exact up to machine precision.
Using our proof-of-principle implementation, we are able to analyze repeater chains of thousands of segments for some parameter regimes.
The algorithms thus serve as useful tools for the analysis of large quantum repeater chain protocols and topologies of the future quantum internet.

\end{abstract}

\IEEEpeerreviewmaketitle

\section{Introduction}
\IEEEPARstart{T}{he} quantum internet is the vision of a global network, running parallel to our current internet, that will enable the transmission of quantum information between arbitrary points on earth \cite{kimble2008quantum,wehner2018quantum}.
The reasons to investigate such a futuristic scenario is that quantum communication enable the implementation
 of applications beyond the reach of their classical counterparts \cite{wehner2018quantum}.
Examples of these tasks range from secure key distribution and communications \cite{bennett1984quantum, ekert1991quantum}, clock synchronization \cite{josza2000quantum}, distributed sensing~\cite{ge2018distributed, zhuang2018distributed}
 and secure delegated quantum computing \cite{childs2005secure} to extending the baseline of telescopes \cite{kellerer2014quantum}.
One of the key elements to enable these applications is the distribution of entanglement between remote parties.
However, the transmission of long-distance long-lived entanglement remains an open experimental challenge.
The main problem to overcome are the losses in the physical transmission medium, typically glass fibre or free space.
Although the impossibility of copying quantum information \cite{wootters1982single} renders signal amplification impossible, it is still possible to reach long distances by means of a chain of intermediate nodes, known as quantum repeaters \cite{briegel1998quantum}, between sender and receiver.
Here, we aim at fully characterizing the behavior of an important class of entanglement distribution protocols over repeater chains as a tool for the analysis of quantum networks.

The key idea behind quantum repeater protocols is to divide the distance separating the two distant parties in a number of segments connected via a quantum repeater.
At these points both losses and errors can be tackled.
A large number of repeater protocols have been proposed \cite{azuma2015all,bernardes2011rate,borregaard2015long,borregaard2016scalable,borregaard2019quantum,brask2008memory,briegel1998quantum,buterakos2017deterministic,collins2007multiplexed,duan2001long,fowler2010surface,guha2015rate,jiang2007fast,jiang2009quantum,munro2010quantum,muralidharan2014ultrafast,muralidharan2016optimal,muralidharan2017overcoming,pant2017rate,santra2018quantum,simon2007quantum,van2009system} and to a large extent they can be classified \cite{munro2015inside,muralidharan2016optimal} depending on whether or not they use error correction codes to handle these issues.
In the absence of coding, losses can be dealt with via heralded entanglement generation and errors via entanglement distillation \cite{duer2007entanglement,gisin1996hidden,bennett1996purification,deutsch1996quantum,bennett1996mixed,campbell2008measurement,rozpkedek2018optimizing,fang2019non}.
Here, we will focus our interest in this type of protocols as their implementation is closer to experimental reach.

Existing analytical work is mostly aimed at estimating the mean waiting time or fidelity (see also \cite{abruzzo2013quantum, guha2015rate, khatri2019practical} for other figures of merit).
Some of this work builds on an approximation of the mean waiting time under the small-probability assumption \cite{jiang2007fast, simon2007quantum, brask2008memory, sangouard2011quantum},
while for a small number of segments or for some protocols it is possible to compute the waiting time probability distribution exactly \cite{bernardes2011rate,khatri2019practical,
shchukin2017waiting,santra2018quantum,vinay2019statistical}.
However, depending on the application different statistics become relevant. For instance, in the presence of decoherence, one is also interested in the variations around the mean. In order to connect two segments via an intermediate repeater, both segments need to produce an entangled pair. When the first pair in one of the segments is ready, it has to wait until the second segment finalizes, and it decoheres while waiting.
In this context, one may need to discard the entanglement after some maximum amount of time \cite{khatri2019practical,rozpkedek2018parameter,rozpedek2018near-term}.
Entanglement is also used as a resource for implementing non-local gates in distributed quantum computers \cite{cirac1999distributed}. In this context, it is relevant to understand the time it takes to generate a pair of the desired quality with probability larger than some threshold, i.e. in the cumulative distribution.
Here, we undertake the problem of fully characterizing the probability distribution of the waiting time and the associated fidelity to the maximally entangled state.

An algorithm to characterize the full waiting time distribution was first obtained in \cite{shchukin2017waiting}
using Markov chain theory.
Its runtime scales with the number of vertices in the Markov chain, which grows exponentially with the number of repeater segments.
In more recent work, Vinay and Kok show how to improve the runtime using results from complex analysis \cite{vinay2019statistical}.
However, this method still remains exponential in the number of repeater segments.
Here, we provide two algorithms for computing the full distribution of the waiting time and fidelity following the same model as in \cite{shchukin2017waiting}.
Both algorithms are polynomial in the number of segments.
Our main tool is the description of the waiting time and fidelity of the first produced end-to-end link as a recursively defined random variable, in line with the recursive structure of the repeater chain protocol.
The first algorithm is a Monte Carlo algorithm which samples from this random variable, whereas our second algorithm is deterministic and computes the waiting time distribution up to a pre-specified truncation time.
The power of the former algorithm lies in its extendibility: it can be used to analyze refined versions of repeater chain protocols which include intermediate entanglement distillation.
The second algorithm is faster and exact: it computes the probability distribution of the waiting time and corresponding fidelity up to a pre-specified truncation point where the only source of error is machine precision.
The speed of our algorithms allows us to analyze repeater chains with more than a thousand segments for some parameter regimes.

The organization of this work is as follows.
In sec.~\ref{sec:preliminaries}, we introduce notation and the family of repeater protocols under study.
Then, in sec.~\ref{sec:random-variable}, we recursively define the waiting time and fidelity of the generation of a single entangled pair between the end nodes as a random variable.
In sec.~\ref{sec:algorithms}, we provide the two algorithms for computing the probability distribution of this random variable.
We show in sec.~\ref{sec:bounds} how to calculate tighter bounds on the mean waiting time than known in previous work.
Numerical results are given in sec.~\ref{sec:numerical-results}.
In Section~\ref{sec:discussion} we discuss the results obtained and provide an outlook for future research.

\section{Preliminaries \label{sec:preliminaries}}
In this section, we elaborate on the repeater chain protocols we study in this work and explain how we model the quantum repeater hardware.

\subsection{Quantum repeater chains \label{sec:bdcz}}
A quantum repeater chain connects two end points via a series of intermediate nodes.
The goal of the two end points is to share an entangled state of two quantum bits or qubits, the unit of quantum information.
In the entire paper, we refer to both the end points as well as the repeaters as `nodes', so that a repeater chain of $N$ segments has $N + 1$ nodes.
We refer to two-qubit entanglement as a `link' and to a link which spans $n$ hops or segments in the repeater chain as an $n$-hop link.
By end-to-end link, we mean a link between the end nodes of the repeater chain.

In the family of repeater chain protocols we study in this work, nodes are able to perform the following three actions: generate fresh entanglement with adjacent nodes, transform short-range entanglement into long-range entanglement by means of entanglement swapping and increase the quality of links through entanglement distillation.
In this section, we first explain each of these actions in more detail and subsequently elaborate on the family of repeater chain protocols.

First, a node can generate fresh entanglement with each of its adjacent nodes.
We also refer to this single-hop entanglement as `elementary links'.

Furthermore, if the middle node $M$ of a two-segment chain with $A$ and $B$ as end nodes shares a link with node $A$ and another link with node $B$, then $M$ can perform an \textit{entanglement swap} or merely `swap', which results in a shared link between $A$ and $B$.
Performing an entanglement swap is identical to performing quantum teleportation on one part of an entangled pair using another entangled pair\cite{bennett1993teleporting}. 
It consists of two parts: first, node $M$ performs a measurement in the Bell basis on the two qubits it holds, which entangles the two qubits held by $A$ and $B$.
Next, node $M$ sends a classical message to $A$ and $B$ in order to notify them of the outcome of this measurement, which determines the precise state of the entangled pair they hold.
In this work, we model the entanglement swap as an operation which succeeds probabilistically: in the case of failure, both involved entangled pairs are lost.

The last action that a node can perform is entanglement distillation, which transforms two or more low-quality links between the same two nodes and produces a single high-quality link \cite{bennett1996purification, deutsch1996quantum}.
In this work, we use entanglement distillation schemes which are probabilistic: the success or failure of a single entanglement distillation step depends on the joint outcomes on the two sides.
These entanglement distillation schemes consist of two steps: first, both nodes perform local operations including a measurement on their parts of the two links, followed by transmitting the measurement outcome to each other.
The combination of measurement outcomes from the two nodes determines whether the distillation step was successful.
In the case of failure, both involved links are lost.

The repeater chain protocols we study in this work are all based on the seminal work of Briegel et al. \cite{briegel1998quantum}, whose scheme was designed for a repeater chain of $N = W^n$ segments with $n\in \{1, 2, \dots\}$.
For simplicity, we assume $W=2$ in this work.
We distinguish between two versions of the protocol.
The first is \mbox{\genswaponly}, where nodes generate elementary entanglement and transform it into end-to-end entanglement by means of entanglement swaps in a particular order explained below.
The next is \mbox{\ddistillation}, which is identical to the \mbox{\genswaponly} version except for the fact that every $n$-hop link is produced $2^d$ times for some integer $d \geq 1$ and then turned into a single high-quality link by performing entanglement distillation multiple times (more details below).

We start by explaining how the \mbox{\genswaponly} protocol works for two segments and subsequently generalize to $2^n$ segments.
On a chain of two segments, the \mbox{\genswaponly} scheme starts with both end nodes generating a single entangled pair with the repeater node (fig.~\ref{fig:bdcz}(a)).
Once a link is generated, the two involved nodes store the state in memory.
As soon as both pairs have been produced, the repeater node performs an entanglement swap on the two qubits it holds, which probabilistically produces a $2$-hop link between the end nodes.
In the case that the entanglement swap did not succeed, both end nodes will be notified of the failure by the heralding message from the repeater node and subsequently each restart generation of the single-hop entanglement.

In the generalization to repeater chains of $2^n$ segments (fig.~\ref{fig:bdcz}(b) and (c)), the \mbox{\genswaponly} scheme starts with the two-segment repeater scheme as explained above on the first and second segment, on the third and fourth, and so on until segments $2^n -1$ and $2^n$.
Approximately half of the intermediate nodes are thus involved in two instances of the scheme; as soon as both instances have finished generating $2$-hop entanglement, the node will perform an entanglement swap to generate a single $4$-hop entangled link.
In case the entanglement swap fails, all nodes under the span of the 4 hops will start to generate single-hop entanglement again as part of the two-segment scheme.
In general, to produce entanglement that spans $2^{\ell}$ hops, the node that is located precisely in the middle of this span will wait for the production of two $2^{\ell - 1}$-hop links and then perform an entanglement swap (see also fig.~\ref{fig:bdcz}(c)).
The failure of this swap requires to regenerate both $2^{\ell - 1}$-hop links.
We refer to $\ell \in \{0, 1, \dots, n\}$ as the `nesting level' of the protocol, such that single-hop entanglement is produced at the base level $\ell = 0$ and the entanglement swap at level $\ell \geq 1$ transforms two $2^{\ell - 1}$-hop links into a single $2^{\ell}$-hop entangled pair.

In addition to the steps of the \mbox{\genswaponly} version as described above, the original proposal of Briegel et al. included entanglement distillation in order to increase the quality of the input links to each entanglement swap.
In this work, we specifically define a version of the repeater protocol, denoted by \mbox{\ddistillation} for some $d \geq 1$, where $d$ rounds of distillation are performed at every nesting level.
That is, instead of a single link, $2^d$ links are generated at every nesting level.
These links are subsequently used as input to a recurrence distillation scheme: our description of this scheme follows the review work by D\"ur and Briegel \cite{duer2007entanglement}.
In the first step of the recurrence protocol, the $2^d$ links are split up in pairs and used as input to entanglement distillation, which produces $2^{d-1}$ entangled pairs of higher quality.
This process is repeated with the remaining links until only a single link is left, which is then used by the repeater node as input link to the entanglement swap.
Rather than waiting for all $2^d$ links to have been generated before performing the first distillation step, the protocol performs the entanglement distillation as soon as two links are available.
The failure of a distillation step requires the two involved nodes to regenerate the links.
For $d=0$, the \mbox{\ddistillation} scheme is identical to \mbox{\genswaponly} since no distillation is performed.

Generating, distilling and swapping entanglement can in general all be probabilistic operations, which makes the total time it takes to distribute a single entangled pair between the end nodes of a repeater chain a random variable.
We use the notation $T_n$ to refer to the waiting time until a single end-to-end link in a $2^n$-segment repeater chain has been produced.
By $F_n$ we refer to the link's fidelity, a measure of the quality of the state (see sec.~\ref{sec:model}).
Every time the quantities $T_n$ and $F_n$ are used in this work we explicitly state whether they correspond to the waiting time of the \mbox{\genswaponly} version or the \mbox{\ddistillation} version for given $d$.
The goal of this work is to find the joint probability distribution of $T_n$ and $F_n$ for both schemes.

\begin{figure*}[h]
	\includegraphics[width=\textwidth]{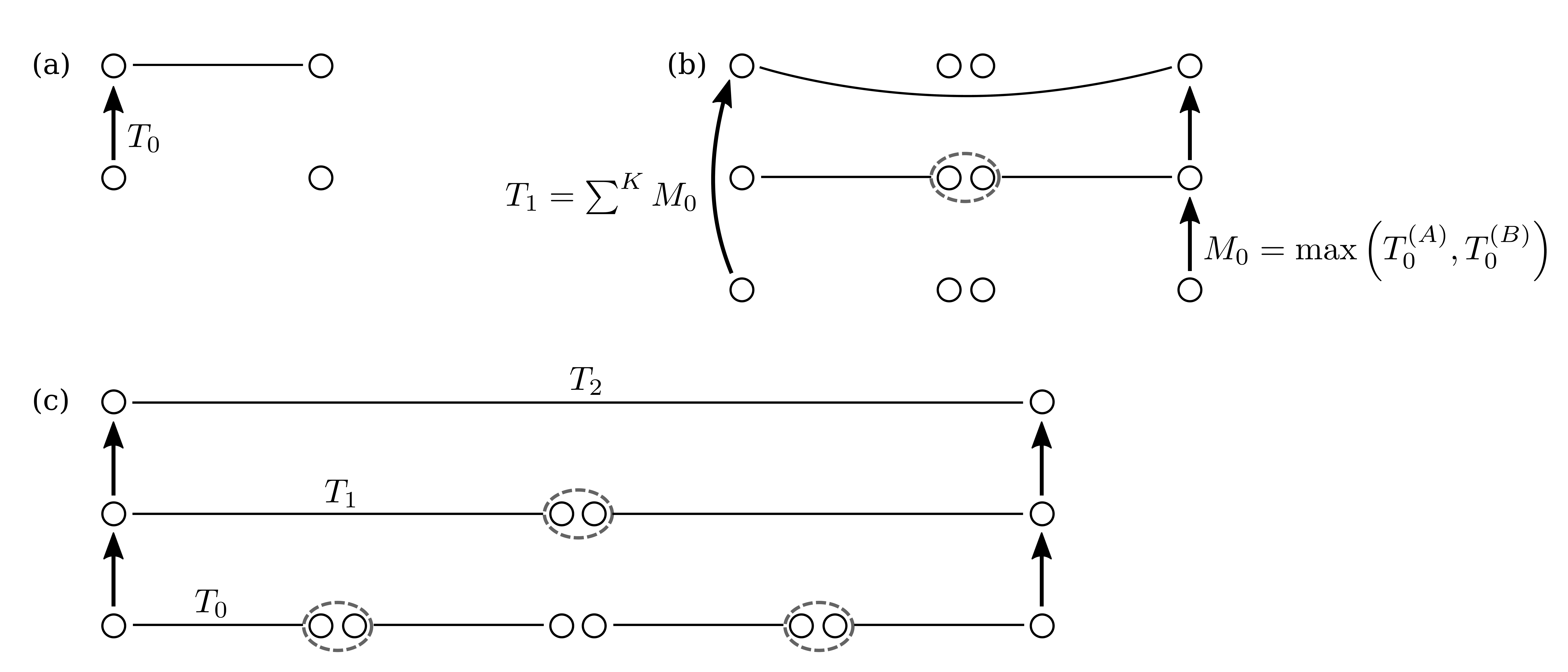}
	\caption{
		The \mbox{\genswaponly} version of the BDCZ protocol \cite{briegel1998quantum} and its completion time $T_n$ as a random variable, where $2^n$ is the number of segments in the repeater chain (see also sec.~\ref{sec:bdcz}).
		(a) For two segments, $T_0$ represents the waiting time for the generation of a single link between two nodes without any intermediate repeater nodes.
		(b) Nested level structure of the protocol over $2^1 = 2$ segments.
		The production of entanglement over two segments first requires the generation of two links, each of which spans a single segment.
		The total time until both links have been generated equals $M_0$, the maximum of their individual generation times $T_0^{(A)}$ and $T_0^{(B)}$, which are independent random variables that are identically distributed (i.i.d.).
		Once the two links have been produced, a probabilistic entanglement swap is performed at both links.
		Failure of the entanglement swap requires the two single-hop links to be regenerated, each of which adds to the total waiting time $T_1$.
		The random variable $K$ corresponds to the number of failing entanglement swaps up to and including the first successful swap.
		In this work we assume that $K$ follows a geometric distribution with parameter $\pswap$ (see sec.~\ref{sec:Tn}).
		(c) A link that spans $2^n$ segments is produced in a nested fashion, where at each nesting level two links are produced and subsequently swapped.
		\label{fig:bdcz}
		}
\end{figure*}

\subsection{Model \label{sec:model}}
In the quantum repeater protocols we study (see sec.~\ref{sec:bdcz}), nodes can generate, store, distill and swap entangled links. We show here how we model each of these four operations.

For the generation of single-hop entanglement between two adjacent nodes, we choose generation schemes which perform heralded attempts of fixed duration $L / c$ where $c$ is the speed of light and $L$ is the distance over which entanglement is generated \cite{munro2015inside}.
In this work we study the topology where all nodes are equally spaced with distance $L=L_0$.

We model entanglement generation to succeed with a fixed probability $0 < \pgen \leq 1$.
For simplicity, we also assume that the success probability $\pgen$ is identical for all pairs of adjacent nodes.
This implies independence between different entanglement generation attempts, i.e. the success or failure of a previous attempt has no influence on future attempts.

The first step of the entanglement swapping, the Bell-state measurement, is modelled as a probabilistic operation with fixed success probability $0 < \pswap \leq 1$ which is identical for all nodes.
This success probability is independent of the state of the qubits that it acts upon.
For simplicity, we assume that the duration of the Bell-state measurement is negligible.
The Bell-state measurement is followed by a classical heralding signal to notify the nodes holding the other sides of the pair whether the Bell-state measurement was successful.
An entanglement swap on two $2^n$-hop links thus takes $2^n \cdot L_0 / c$ time.
Although our algorithms can account for this communication time (see sec.~\ref{sec:comm-time}), we will assume this time to be negligible in most of this work.

The fidelity $F(\rho, \sigma) \in [0, 1]$ between two quantum states on the same number of qubits, represented as density matrices $\rho$ and $\sigma$, is a measure of their closeness, defined as
\[
	F(\rho, \sigma) := \Tr(\sqrt{\sqrt{\rho}\sigma\sqrt{\rho}})^2
	\]
which implies that $F(\rho, \sigma) = 1$ precisely if $\rho = \sigma$.
By Bell-state fidelity, we mean the fidelity between $\sigma$ and \mbox{$\rho = \dyad{\Phi^+}$} where $\ket{\Phi^+} = (\ket{00} + \ket{11}) / \sqrt{2}$ is a Bell state.

We assume that the single-hop entangled states that are generated are two-qubit Werner states parameterized by a single parameter $0 \leq \werner_0 \leq 1$ \cite{werner1989quantum}:
\begin{equation}
\label{eq:werner-states}
	\rho(\werner_0) = \werner_0 \dyad{\Phi^+} + (1 - \werner_0) \frac{\unit_4}{4}	
\end{equation}
where $\unit_4 / 4 = (\dyad{00} + \dyad{01} + \dyad{10} + \dyad{11}) / 4$ is the maximally-mixed state on two qubits.
A straightforward computation shows that the fidelity between $\rho(\werner)$ and the Bell state $\ket{\Phi^+}$ equals
\begin{equation}
\label{eq:fidelity-werner-states}
	F\left(\rho(\werner), \dyad{\Phi^+}\right) = \expval{\rho(\werner)}{\Phi^+} = (1 + 3\werner) / 4
	.
\end{equation}

Quantum states that are stored in the memories decohere over time with the following noise: a Werner state $\rho(\werner)$ residing in memory for a time $\Delta t$ will transform into the Werner state $\rho(\werner_{\textnormal{decayed}})$ with
\begin{equation}
\label{eq:decay}
	\werner_{\textnormal{decayed}} = \werner \cdot e^{-\Delta t / T_{\textnormal{coh}}}
\end{equation}
where $T_{\textnormal{coh}}$ is the joint coherence time of the two quantum memories holding the qubits.
We assume that access to a quantum memory is on-demand, i.e. the quantum states can be stored and retrieved at any time and moreover there is no fidelity penalty associated with such memory access. 

A successful entanglement swap acting on two Werner states $\rho(\werner)$ and $\rho(\werner')$ will produce the Werner state
\begin{equation}
\label{eq:werner-swap}
\rho_{\textnormal{swap}} = \rho(\werner\cdot \werner')
.
\end{equation}
We assume that the Bell-state measurement and the local operations that the entanglement swap consists of are noiseless and instantaneous.

As base for entanglement distillation, we use the BBPSSW-scheme \cite{bennett1996purification}.
We modify it slightly by bringing the output state back into Werner form.
The last step does not change the Bell-state fidelity of the output state.
If two Werner states with parameters $\werner_A$ and $\werner_B$ are used as input to entanglement distillation, both the output Werner parameter $\werner_{\textnormal{dist}}$ and the success probability $p_{\textnormal{dist}}$ depend on the Werner parameters $\werner_A$ and $\werner_B$ of the states it acts upon (see appendix~\ref{app:distillation}):
\begin{eqnarray}
	\label{eq:werner-distillation}
	\werner_{\textnormal{dist}}
	(\werner_A, \werner_B)
	&=&
	\frac{1 + \werner_A + \werner_B + 5 \werner_A \werner_B}{6p_{\textnormal{dist}}} - \frac{1}{3}
	\\
	\label{eq:pdist}
	p_{\textnormal{dist}} 
	(\werner_A, \werner_B)
	&=&
	(1 + \werner_A \werner_B) / 2
	.
\end{eqnarray}
The two nodes involved in distillation on two $2^n$-hop states send their individual measurement outcomes to each other, which takes $2^n \cdot (L_0 /c)$ time but we will assume this time to be negligible for simplicity.
We also assume that the duration of the local operations needed for the distillation is negligible.

\subsection{Notation: random variables \label{sec:notation}}
In this section, we fix notation on random variables and operations on them.

Most random variables in this paper are discrete with (a subset of) the nonnegative integers as domain.
Let $X$ be such a random variable, then its probability distribution function $p_X: x \mapsto \Pr(X = x)$ describes the probability that its outcome will be $x \in \{0, 1, 2, \dots\}$.
Equivalently, $X$ is described by its cumulative distribution function \mbox{$\Pr(X \leq x) = \sum_{y=0}^x \Pr(X = y)$}, which is transformed to the probability distribution function as \mbox{$\Pr(X = x) = \Pr(X \leq x) - \Pr(X \leq x - 1)$}.
Two random variables $X$ and $Y$ are independent if \mbox{$\Pr(X = x \textnormal{ and } Y=y) = \Pr(X=x) \cdot \Pr(Y=y)$} for all $x$ and $y$ in the domain.
By a `copy' of $X$, we mean a fresh random variable which is independent from $X$ and identically distributed (i.i.d.).
We will denote a copy by a superscript in parentheses.
For example, $X^{(1)}, X^{(142)}$ and $X^{(A)} $ are all copies of $X$.

The mean of $X$ is denoted by \mbox{$\mean{X} = \sum_{x=0}^{\infty} \Pr(X=x) \cdot x$} and can equivalently be computed as \mbox{$\mean{X} = \sum_{x=1}^{\infty} \Pr(X \geq x)$}.
If $f$ is a function which takes two nonnegative integers as input, then the random variable $f(X, Y)$ has probability distribution function 
\[
	\Pr(f(X, Y) = z) := \sum_{\substack{x=0, y=0:\\ f(x, y) = z}}^{\infty} \Pr(X = x \textnormal{ and } Y=y)
	.
\]
An example of such a function is addition.
Define $Z := X + Y$ where $X$ and $Y$ are independent, then the probability distribution $p_Z$ of $Z$ is given by the convolution of the distributions $p_X$ and $p_Y$, denoted as \mbox{$p_Z = p_X * p_Y$}, which means \cite{feller1957introduction}
\[
	p_Z (z) = \Pr(Z = z) = \sum_{x=0}^{z} p_X(x) \cdot p_Y(z - x)
	.
\]
The convolution operator $*$ is associative $\left((a * b) * c = a * (b * c)\right)$ and thus writing $a * b * c$ is well-defined, for functions $a, b, c$ from the nonnegative integers to the real numbers.
In general, the probability distribution of sums of independent random variables equals the convolutions of their individual probability distribution functions.

\section{Recursive expressions for the waiting time and fidelity as a random variable \label{sec:random-variable}}
In this section, we derive expressions for the waiting time and fidelity of the first generated end-to-end link in the \mbox{\genswaponly} repeater chain protocol.
First, in sec.~\ref{sec:Tn}, we derive a recursive definition for the random variable $T_n$, which represents the waiting time in a $2^n$-segment repeater chain.
Section~\ref{sec:TnWn} is devoted to extending this definition to the Werner parameter $W_n$ of the pair, which stands in one-to-one correspondence to its fidelity $F_n$ using eq.~\eqref{eq:fidelity-werner-states}:
\begin{equation}
	\label{eq:Fn}
	F_n = \left( 1 + 3 W_n\right) / 4
	.
\end{equation}
In sec.~\ref{sec:comm-time}, we show how to include the communication time after the entanglement swap and in sec.~\ref{sec:random-variable-extended}, we extend the analysis of the waiting time and Werner parameter in the \mbox{\genswaponly} protocol to the \mbox{\ddistillation} scheme.

\subsection{Recursive expression for the waiting time in the \mbox{\genswaponly} protocol \label{sec:Tn}}
In the following, we will derive a recursive expression for the waiting time $T_n$ of a \mbox{\genswaponly} repeater chain of $2^n$ segments (see also fig.~\ref{fig:bdcz}).

Before stating the expression, let us note that all three operations in the repeater chain protocols we study in this work, entanglement generation over a single hop, distillation and swapping, take a duration that is a multiple of $L_0 / c$, the time to send information over a single segment (see sec.~\ref{sec:model} for our assumptions on the duration of operations).
For this reason, it is common to denote the waiting time in discrete units of $L_0 / c$, which is a convention we comply with for $T_n$.

Let us first state the description of $T_n$ before explaining it.

\begin{center}
\fbox{
\begin{minipage}{0.4\textwidth}
	\begin{center}
\textbf{Waiting time in the \mbox{\genswaponly} protocol}\\
	\end{center}
\label{def:Tn}
	We recursively describe the random variable $T_n$ that represents the waiting time until the first end-to-end link in a $2^n$-segment \mbox{\genswaponly} repeater chain is generated, for $n \in \{0, 1, \dots\}$.
The waiting time $T_0$ for generating point-to-point entanglement follows a geometric distribution with parameter $\pgen$.
	At the recursive step, the waiting time is given as the geometric compound sum
	\begin{equation}
	\label{eq:waiting-time}
		T_{n + 1} := \sum_{j=1}^{K_{n}} M_{n}^{(j)}
	\end{equation}
	where $M_{n}$ is an auxiliary random variable given by
	\begin{equation}
	\label{eq:max}
		M_{n} := g_{\textnormal{T}}\left(T_{n}^{(A)}, T_{n}^{(B)}\right)
	\end{equation}
	and the function $g_{\textnormal{T}}$ is defined as
	\begin{equation}
	\label{eq:gT}
	g_{\textnormal{T}} (t_A, t_B) := \max\{t_A, t_B\}
	.
	\end{equation}
	The sum in eq.~\eqref{eq:waiting-time} is taken over the number of entanglement swaps $K_{n}$ until the first success, which is geometrically distributed with parameter $\pswap$ for every $n$.
	See fig.~\ref{fig:bdcz} for a depiction of $T_n$ and $K_n$.
\end{minipage}
}
	\vspace{\baselineskip}
\end{center}

Let us now elaborate on each of the steps in the expression of $T_n$.

We start with the base case $T_0$, the waiting time for the generation of elementary entanglement.
Since we model the generation of single-hop entanglement by attempts which succeed with a fixed probability $\pgen$ (see sec.~\ref{sec:model}), the waiting time $T_0$ is a discrete random variable (in units of $L_0 / c$) which follows a geometric distribution with probability distribution given by $\Pr(T_0 = t) = \pgen (1 - \pgen)^{t - 1}$ for $t \in \{1, 2, 3, \dots\}$.
For what follows, it will be more convenient to specify $T_0$ by its cumulative distribution function
\begin{equation}
\label{eq:T0}
\Pr(T_0 \leq t) = 1 - (1 - \pgen)^t
.
\end{equation}

Let us now assume that we have found an expression for $T_{n}$ and we want to construct $T_{n + 1}$.
In order to perform the entanglement swap to produce a single $2^{n+1}$-hop link, a node needs to wait for the production of two $2^n$-hop links, one on each side.
Denote the waiting time for one of the pairs by $T_n^{(A)}$ and the other by $T_n^{(B)}$, both of which are i.i.d. with $T_n$.
The time until both pairs are available is now given by $M_n := \max(T_n^{(A)}, T_n^{(B)})$ which is distributed according to
\begin{eqnarray}
	\nonumber
	\Pr(M_n \leq t)
	&=&
	\Pr(T_n^{(A)} \leq t \textnormal{ and } T_n^{(B)} \leq t)
	\\
	&=& \Pr(T_n \leq t)^2
\label{eq:M0}
\end{eqnarray}
where the last equality follows from the fact that $T_n^{(A)}, T_n^{(B)}$ and $T_n$ are pairwise i.i.d.
Since we assume that both the duration of the Bell-state measurement and the communication time of the heralding signal after the entanglement swap are negligible (see sec.~\ref{sec:model}), $M_n$ is also the time at which the entanglement swap ends.
We will drop the assumption on negligible communication time in sec.~\ref{sec:comm-time}.

In order to find the relation between $M_n$ and $T_{n+1}$, first note that the number of swaps $K_n$ at level $n$ until the first successful swap follows a geometric distribution with parameter $\pswap$.
This is a direct consequence of our choice to model the success probability $\pswap$ to be independent of the state of the two input links (see sec.~\ref{sec:model}).
Next, recall that the two input links of a failing entanglement swap are lost and need to be regenerated.
The regeneration of fresh entanglement after each failing entanglement swap adds to the waiting time.
Thus, $T_{n + 1}$ is a \textit{compound random variable}: it is the sum of $K_n$ copies of $M_n$.
Since the number of entanglement swaps $K_n$ is geometrically distributed, we say that $T_{n + 1}$ is a \textit{geometric compound sum} of $K_n$ copies of $M_n$.
To be precise, we write
\begin{equation}
\label{eq:T1}
	T_{n + 1} = \sum_{k=1}^{K_n} M_n^{(k)}
\end{equation}
which means that the probability distribution of the waiting time $T_{n + 1}$ is computed as the marginal of the waiting time conditioned on a fixed number of swaps:
\begin{equation*}
\label{eq:T1-pmf}
	\Pr(T_{n + 1} = t)
	= \sum_{k=1}^{\infty} \Pr(K_n=k) \cdot
	\Pr\left[ \left(\sum_{j=1}^{k} M_n^{(j)}\right) = t\right]
\end{equation*}
where the $M_n^{(j)}$ are copies of $M_{n}$.

The waiting time $T_n$ is the same quantity as was studied by Shchukin et al. \cite{shchukin2017waiting}.
Indeed, in sec.~\ref{sec:numerical-results}, we show that our algorithms for computing the probability distribution of $T_n$ recover their numerical results.

\subsection{Joint recursive expression of waiting time and Werner parameter for the \mbox{\genswaponly} protocol \label{sec:TnWn}}

In this section, we extend the expression of the waiting time for the first end-to-end link produced using the \mbox{\genswaponly} protocol with the link's state.
To be precise, we give a recursive expression for the waiting time $T_n$ and Werner parameter $W_n$ of this state, which is well-defined since all states that the \mbox{\genswaponly} repeater chain protocol holds at any time during its execution are Werner states.
The latter statement is a direct consequence of the fact that in our modeling, all operations in the \mbox{\genswaponly} protocol only output Werner states: we choose to model the generated single-hop entanglement as Werner states and furthermore the class of Werner states is invariant under memory errors and entanglement swaps (see sec.~\ref{sec:model}).
The fidelity $F_n$ of the first end-to-end state on $2^n$ segments can be computed from its Werner parameter using eq.~\eqref{eq:Fn}.

We express the waiting time and Werner parameter as a joint random variable $(T_n, W_n)$.
Describing the two as a tuple allows us to capture the fact that the Werner parameter of a link depends on the time it was produced at.
In sec.~\ref{sec:Tn}, we found that the failure of multiple swapping attempts corresponds to the sum of their waiting times.
In order to extend this description to the tuple of waiting time and Werner parameter, we define the \textit{forgetting sum} $\sumforget{}{}$ on sequences of tuples \mbox{$\{(x_j, y_j) | 1 \leq j \leq m\}$} for some \mbox{$m \in \{1, 2, \dots \}$} as
\begin{equation}
\label{eq:forgetting-sum}
	\sumforget{j=1}{m} (x_j, y_j) := \left( \sum_{j=1}^m x_j, y_m\right)
	.
\end{equation}
In analogy to the geometric compound sum from eq.~\eqref{eq:T1}, we define the \textit{geometric compound forgetting sum} \mbox{$(X', Y') := \sumforget{j=1}{K} (X, Y)$}, which formally means
\begin{eqnarray*}
\label{eq:T1-pmf}
	&&\Pr(X' = x \textnormal{ and } Y' = y)
	\\
	&&=\sum_{k=1}^{\infty} p (1 - p)^{k-1} \cdot
	\Pr\left( \sumforget{j=1}{k} (X, Y)^{(j)} = (x, y)\right)
\end{eqnarray*}
where $X$ and $Y$ and their primed version are random variables, and $K$ is a geometrically distributed random variable with parameter $p$.

Making use of the compound forgetting sum, we give the expression for the joint random variable of waiting time $T_n$ and Werner parameter $W_n$.
\begin{center}
\fbox{
\begin{minipage}{0.4\textwidth}
	\begin{center}
\textbf{Waiting time and Werner parameter in the \mbox{\genswaponly} protocol}\\
	\end{center}
	The joint random variable $(T_n, W_n)$ is defined as follows.
	The waiting time $T_0$ is the same as in sec.~\ref{sec:Tn} and \mbox{$\Pr(W_0 = \werner_0) = 1$} where \mbox{$\werner_0\in [0, 1]$} is some pre-specified constant that determines the state of the single-hop entanglement that is produced between adjacent nodes.
	At the recursive step, the waiting time and Werner parameter are given by the geometric compound forgetting sum
	\begin{equation}
		\label{eq:TnWn}
		(T_{n + 1}, W_{n+1}) := \sumforget{k=1}{K_{n}} (M_{n}, V_{n})^{(k)}
	\end{equation}
	where, as in sec.~\ref{sec:Tn}, $K_{n}$ follows a geometric distribution with parameter $\pswap$.
	The auxiliary joint random variable $(M_{n}, V_{n})$ is defined as
	\begin{equation}
		\label{eq:max-TnWn}
		(M_{n}, V_{n}) := g\left((T_{n}, W_{n})^{(A)},  (T_{n}, W_{n})^{(B)}\right)
		.
	\end{equation}
	The function $g$ is given by
	\begin{eqnarray}
		\label{eq:g}
		&&g((t_A, \werner_A), (t_B, \werner_B)) :=\\
		&&\quad (g_{\textnormal{T}}(t_A, t_B), g_{\textnormal{W}}((t_A, \werner_A), (t_B, \werner_B)))
		\nonumber
	\end{eqnarray}
	where $g_{\textnormal{T}}$ is defined in eq.~\ref{eq:gT} and
	\begin{equation}
	\label{eq:gW}
	g_{\textnormal{W}} ((t_A, \werner_A), (t_B, \werner_B)) := \werner_A \cdot \werner_B \cdot e^{-|t_A - t_B|/ T_{\textnormal{coh}}}
	\end{equation}
	with $T_{\textnormal{coh}}$ the quantum memory coherence time as described in sec.~\ref{sec:model}.
\end{minipage}
}
\end{center}

We now explain the above expressions.
For a single segment ($n=0$), the waiting time and Werner parameter are uncorrelated because we model the attempts at generating single-hop entanglement to be independent and to each take equally long (see sec.~\ref{sec:model}).
At the recursive step, an entanglement swap which produces $2^{n + 1}$-hop entanglement requires the generation of two $2^{n}$-hop links.
The expression for the waiting time $T_{n + 1}$ is identical to eq.~\eqref{eq:waiting-time} in sec.~\ref{sec:Tn}.
In order to argue that eq.~\eqref{eq:TnWn} also gives the correct expression for $W_{n + 1}$, we first show that the Werner parameter of the output link of an entanglement swap is given by $V_n$ in eq.~\eqref{eq:max-TnWn}, provided the swap succeeded.
Since $M_n$ as defined in eq.~\eqref{eq:max-TnWn} is identical to its expression in eq.~\eqref{eq:max} in sec.~\ref{sec:Tn}, we only need to argue why $g_{\textnormal{W}}$ in eq.~\eqref{eq:gW} correctly computes the Werner parameter of the output link after an entanglement swap.

In order to do so, denote by $A$ and $B$ the input links to the entanglement swap and denote by $(t_A, \werner_A)$ and $(t_B, \werner_B)$ their respective delivery times and Werner parameters.
Without loss of generality, choose $t_A \geq t_B$, i.e. link $A$ is produced after link $B$.
Link $A$ is produced last, so the entanglement swap will be performed directly after its generation and hence link $A$ will enter the entanglement swap with Werner parameter $w_A$.
Link $B$ is produced earliest and will therefore decohere until production of link $A$.
It follows from eq.~\eqref{eq:decay} that $B$'s Werner parameter immediately before the swap equals
	\begin{equation}
		\label{eq:werner-decayed-before-swap}
	\werner_B' = \werner_B \cdot e^{-|t_A - t_B| / T_{\textnormal{coh}}}
	.
	\end{equation}
Once two links have been delivered, the entanglement swap would produce the $2^{n + 1}$-hop state with Werner parameter
\begin{equation}
	\label{eq:werner-would-be}
\werner_A \cdot \werner_B'
\end{equation}
as in eq.~\eqref{eq:werner-swap}, provided the swap is successful.
Combining eqs.~\eqref{eq:werner-decayed-before-swap} and \eqref{eq:werner-would-be} yields the definition of $g_{\textnormal{W}}$ in eq. \eqref{eq:gW}.

Note that in the definition of $g_{\textnormal{W}}$ in eq.~\eqref{eq:gW} we used the same assumption on the duration of the entanglement swap as in sec.~\ref{sec:Tn}, i.e. that both the Bell-state measurement and the subsequent communication time are negligible (see also sec.~\ref{sec:model}).
This implies that $V_{n}$ in eq.~\eqref{eq:max-TnWn} expresses the Werner parameter of the produced $2^{n + 1}$-hop link in case the swap is successful.
We treat the case of nonzero communication time in sec.~\ref{sec:comm-time}.

The last step in finding the Werner parameter $W_{n+1}$ in eq.~\eqref{eq:TnWn} is to bridge the gap with $(M_n, V_n)$ from eq.~\eqref{eq:max-TnWn}.
If the entanglement swap fails, then the $2^{n + 1}$-hop link with its Werner parameter in eq.~\eqref{eq:werner-would-be} will never be produced since both initial $2^{n}$-hop entangled pairs are lost.
Instead, two fresh $2^{n}$-hop links will be generated.
In order to find how the Werner parameter on level $n+1$ is expressed as a function of the waiting times and Werner parameters at level $n$, consider a sequence $(m_j, v_j)$ of waiting times $m_j$ and Werner parameters $v_j$, where $j$ runs from $1$ to the first successful swap $k$.
The $m_j$ correspond to the waiting time until the end of the entanglement swap that transforms two $2^{n}$-hop links into a single $2^{n + 1}$-hop link and the $v_j$ to the output link's Werner parameter if the swap were successful.
We have found in sec.~\ref{sec:Tn} that the total waiting time is given by $\sum_{j=1}^k m_j$, the sum of the duration of the production of the lost pairs (see eq.~\eqref{eq:waiting-time}).
Note, however, that the Werner parameter of the $2^{n + 1}$-hop link is only influenced by the links that the \textit{successful} entanglement swap acted upon.
Since the entanglement swaps are performed until the first successful one, the output link is the last produced link and therefore its Werner parameter equals $v_k$.
We thus find that the waiting time $t_{\textnormal{final}}$ of the first $2^{n + 1}$-hop link and its Werner parameter $w_{\textnormal{final}}$ are given by the forgetting sum from eq.~\eqref{eq:forgetting-sum}:
\[
	(t_{\textnormal{final}}, w_{\textnormal{final}}) = \left( \sum_{j=1}^k m_j, v_k\right) = \sumforget{j=1}{k} (m_j, v_j)
	.
	\]
Taking into account that the number of swaps $k$ that need to be performed until the first successful one is an instance of the random variable $K_n$, we arrive at the full recursive expression for the waiting time and Werner parameter at level $n +1$ as given in eq.~\eqref{eq:TnWn}.

It is not hard to see that the projection $(T_{n}, W_{n}) \mapsto T_{n}$ recovers the definition of waiting time from \ref{def:Tn}.
Indeed, following the recursive definition of $(T_{n}, W_{n})$ in eqs.~\eqref{eq:TnWn}-\eqref{eq:gW}, the waiting time $T_{n}$ is not affected by the Werner parameters $W_{\ell}$ at lower nesting levels $\ell < n$.

\subsection{Including communication time \label{sec:comm-time}}
While deriving the expressions for waiting time and Werner parameter of the first produced end-to-end link in secs.~\ref{sec:Tn} and \ref{sec:TnWn}, we have explicitly assumed that the total time the entanglement swap takes is negligible.
Here, we include the communication time of the heralding signal from the entanglement swap into the expressions for $M_n$ and $V_n$ (eqs.~\eqref{eq:max} and \eqref{eq:max-TnWn}), which represent the waiting time and Werner parameter directly after the entanglement swap if it were successful.
This communication time equals $2^{n}$ time steps (in units of $L_0 / c$) for a swap that transforms two $2^{n}$-hop links into a single $2^{n+1}$-hop link (see sec.~\ref{sec:model}).
The expressions for $M_n$ and $V_n$ are modified by replacing $g_{\textnormal{T}}$ in eq.~\ref{eq:gT} by
\begin{equation}
\label{eq:gT-with-communication-time}
	g_{\textnormal{T}}^{n} (t_A, t_B) := g_{\textnormal{T}} (t_A, t_B) + 2^{n}
.
\end{equation}
and replacing $g_{\textnormal{W}}$ from eq.~\eqref{eq:gW} by
\begin{eqnarray}
\nonumber
	&&g_{\textnormal{W}}^{n} ((t_A, w_A), (t_B, w_B))
	\\
	&&\qquad := g_{\textnormal{W}} ((t_A, w_A), (t_B, w_B)) \cdot e^{-2^{n} / T_{\textnormal{coh}}}
\label{eq:gW-with-communication-time}
.
\end{eqnarray}
Equation~\eqref{eq:gT-with-communication-time} expresses that the entanglement swap takes $2^{n}$ timesteps longer, while eq.~\eqref{eq:gW-with-communication-time} captures the decoherence of the state during the communication time of the entanglement swap, following eq.~\eqref{eq:decay}.

\subsection{Waiting time and Werner parameter for the \mbox{\ddistillation} protocol \label{sec:random-variable-extended}}
In this section, we sketch how to extend the expression of the waiting time $T_n$ and Werner parameter $W_n$ from secs.~\ref{sec:Tn}-\ref{sec:comm-time} to the case of the \mbox{\ddistillation} repeater protocol presented in sec.~\ref{sec:bdcz}.
Recall that the \mbox{\ddistillation} protocol is identical to the \mbox{\genswaponly} protocol except for the fact that each entanglement swap is performed on the output of a recurrence distillation scheme with $d$ nesting levels.
By a $d'$-distilled $2^n$-hop link we denote a $2^n$-hop link which is the result of successful entanglement distillation on two $(d'-1)$-distilled $2^n$-hop links and by a $0$-distilled $2^n$-hop link we mean a link that is the result of a successful entanglement swap on two $2^n$-hop links.
Thus, every entanglement swap in the \mbox{\ddistillation} protocol is performed on $d$-distilled links only.

Note that at every level of the nested swapping, there are $d$ levels of nested distillation.
To tackle the `double nesting' we modify the waiting time in the \mbox{\genswaponly} protocol by splitting up the tuple of random variables $(T_n, W_n)$ in eq.~\eqref{eq:TnWn}, which represents the waiting time and Werner parameter at level $n$, into \mbox{$d + 1$} tuples of random variables $(T_n^{d'}, W_n^{d'})$ for $d' \in \{0, 1, \dots, d\}$.
The random variable $T_n^{d'}$ corresponds to the waiting time until the end of the first successful distillation attempt on two $d'$-distilled $2^n$-hop links, and $W_n^{d'}$ to the link's Werner parameter.

We first analyze the recurrence distillation protocol at a single swapping nesting level and subsequently tie this analysis in with the nested swapping structure.

If we fix the nesting level $n$, we can straightforwardly apply the analysis of sec.~\ref{sec:TnWn} to the nested distillation.
First, we define $(M_n^{d'}, V_n^{d'})$, which characterizes a link after a single distillation attempt on two $2^n$-hop $d'$-distilled links in case the attempt is successful.
This joint random variable is the analogue of $(M_n, V_n)$ from eq.~\eqref{eq:max-TnWn}, which has the same interpretation but in this case for a swapping attempt.
The analysis resulting in eq.~\eqref{eq:max-TnWn} carries over and yields
\begin{equation}
	\label{eq:Mnd}
	(M_{n}^{d'}, V_{n}^{d'}) := g_{\textnormal{D}} \left((T_{n}^{d'}, W_{n}^{d'})^{(A)},  (T_{n}^{d'}, W_{n}^{d'})^{(B)}\right)
	.
\end{equation}
where $g_{\textnormal{D}}$ is the analogue of $g$ in eq.~\eqref{eq:g} and describes how two input links are transformed into one high-quality link by a successful distillation step:
\begin{eqnarray}
	g_{\textnormal{D}} ((t_A, w_A), (t_B, w_B))
	=
	\nonumber
	\left(g_{\textnormal{T}}(t_A, w_A)
	,
	w
	\right)
	\label{eq:gD}
\end{eqnarray}
where
\[
w: = 
	\begin{cases}
	w_{\textnormal{dist}} \left(w_A \cdot e^{-|t_A - t_B|/T_{\textnormal{coh}}}, w_B\right)
		& \text{ if $t_A \leq t_B$}
	\\
	w_{\textnormal{dist}} \left(w_A, w_B \cdot e^{-|t_A - t_B|/T_{\textnormal{coh}}}\right)
		& \text{ if $t_A > t_B$}
	\end{cases}
	\]
and $w_{\textnormal{dist}}$ is given in eq.~\eqref{eq:werner-distillation}.
The function $g_{\textnormal{D}}$ outputs a tuple of waiting time and Werner parameter of the output state after distillation.
The waiting time requires two links to be generated and is thus given by $g_{\textnormal{T}}$ in eq.~\eqref{eq:gT}.
The Werner parameter equals the Werner parameter of distillation as given by $w_{\textnormal{dist}}$ in eq.~\eqref{eq:werner-distillation} on the two input links, of which the earlier suffered decoherence as given in eq.~\eqref{eq:decay}.

The random variables $(T_n^{d'}, W_n^{d'})$ correspond to the waiting time and Werner parameter after the first successful distillation attempt on two $d'$-distilled $2^n$-hop links, so in line with the analysis leading to eq.~\eqref{eq:TnWn} we obtain
\begin{equation}
	\label{eq:Tnd}
	(T_n^{d'+1}, W_n^{d' + 1}) = \sumforget{j=1}{\mathcal{D}_n^{d'}} \left(M_n^{d'}, V_n^{d'}\right)^{(j)}
	.
\end{equation}
The random variable $\mathcal{D}_n^{d'}$ corresponds to the number of distillation attempts with two $d'$-distilled $2^n$-hop links as input, up to and including the first successful attempt.
It is the analogue of $K_n$ in eq.~\eqref{eq:TnWn}, the number of swap attempts until the first success.

At this point, we have an expression for $(T_n^d, W_n^d)$, the waiting time and Werner parameter of the resulting link after performing a $d$-level recurrence protocol on $0$-distilled input links that each span $2^n$ hops.
Since the recurrence protocol is performed at every swapping nesting level of the \mbox{\ddistillation} protocol, we can insert this expression into our previous analysis using the following two remarks.
First, a $0$-distilled link is the output of an entanglement swap, so $(T_n^0, W_n^0)$ in the \mbox{\ddistillation} scheme takes the role that $(T_n, W_n)$ has in the \mbox{\genswaponly} protocol:

\begin{equation}
	\label{eq:Tn0}
	(T_n^0, W_n^0) = (T_n, W_n)
	.
\end{equation}

Second, since an entanglement swap takes as input two $d$-distilled links, we find that we should replace the definition of $(M_n, V_n)$ in eq.~\eqref{eq:max-TnWn} by
\begin{equation}
	\label{eq:Mn-dist}
	(M_n, V_n) = g\left(\left(T_n^{d}, W_n^{d}\right)^{(A)}, \left(T_n^d, W_n^d\right)^{(B)}\right)
	.
\end{equation}
where $g$ is defined in eq.~\eqref{eq:g}.

We finish this section by remarking that for the \mbox{\ddistillation} protocol, we cannot treat waiting time independently of the Werner parameter of the produced link, as we did for the \mbox{\genswaponly} scheme in sec.~\ref{sec:Tn}.
The reason behind this is the following difference between the nested swaps and the nested distillation: in the former, the success probability $\pswap$ and therefore the number of swaps $K_n$ is independent of the time and state of the produced links, whereas the success probability of entanglement distillation is a function of their states (see eq.~\eqref{eq:pdist}).
Consequently, the summation bound $\mathcal{D}_n^{d'}$ and the Werner parameter $V_n^{d'}$ in the summands $\left(M_n^{d'}, V_n^{d'} \right)$ in eq.~\eqref{eq:Tnd} are correlated.
Therefore, both the waiting time and Werner parameter at any swapping level depend on both waiting time and Werner parameter at the levels below.

\section{Algorithms for computing waiting time and fidelity of the first end-to-end link\label{sec:algorithms}}
In this section, we present two algorithms for determining the probability distribution of the waiting time $T_n$ and average Werner parameter $W_n$ of the first end-to-end link produced by the repeater chain (see sec.~\ref{sec:random-variable}).
The first algorithm is a Monte Carlo algorithm which applies to both families of repeater chain protocols considered in this work: \mbox{\genswaponly} and \mbox{\ddistillation}.
The second algorithm only applies to the \mbox{\genswaponly} protocol and is faster than the first.
We summarize the runtime of the different algorithms presented in this section in table~\ref{table:runtimes}.

\begin{table*}
\begin{center}
	\normalsize
    \begin{tabular}{| l | l | c | c | c |}
    \hline
	    Repeater chain protocol

	    & & Markov-chain-approach & \multicolumn{2}{c|}{Algorithms in this work} \\
	    \qquad\quad (sec.~\ref{sec:bdcz})			 &&\cite{shchukin2017waiting}\cite{vinay2019statistical}& Monte-Carlo & Deterministic
	    \\\hline
	    \multirow{4}{*}{\mbox{\genswaponly}} &
	    Waiting time up to $99\%$ of
	    &&&
	    \\&the cumulative probabilities
	    & $\Theta\left(\textnormal{exp}\left(N\right)\right)$  &  $\mathcal{O}\left(\textnormal{poly}\left(N\right)\right)$   & $\Theta\left(\textnormal{poly}\left(N\right)\right)$
	    \\\cline{2-5}
	    & Waiting time up to
	    &&&\\&fixed truncation time $\ttrunc=1000$
		& $\Theta\left(\textnormal{exp}\left(N\right)\right)$
	    &  $\mathcal{O}\left(\textnormal{poly}\left(N\right)\right)$   & $\Theta\left(\log\left(N\right)\right)$ \\\cline{2-5}
	    &Fidelity  & $\cross$ &  $\mathcal{O}\left(\textnormal{poly}\left(N\right)\right)$   & $\Theta\left(\textnormal{poly}\left(N\right)\right)$
	    \\\hline
	    \mbox{\ddistillation} & Waiting time \& fidelity
	    & $\cross$ & $\mathcal{O}\left(\textnormal{poly}\left(N\right)\right)$ & $\cross$\\\hline
    \end{tabular}
	\caption{
		\label{table:runtimes}
		The time complexity of the algorithms for computing waiting time and fidelity of entanglement distribution through repeater chains as presented in this work compared to existing algorithms.
		The algorithms have exponential (exp) or polynomial (poly) runtime in $N=2^n$, the number of segments in the repeater chain, for $n\in \{1, 2, \dots\}$.
		The Monte Carlo algorithm is a randomized algorithm; its presented runtime is the average runtime.
		The cross ($\cross$) indicates that the algorithm is not present.
		}
\end{center}
\end{table*}

\subsection{First algorithm: Monte Carlo simulation \label{sec:monte-carlo}}
The first algorithm is a randomized function which produces a sample from the probability distribution of the joint random variable $(T_n, W_n)$.
By running the algorithm many times, sufficient statistics can be produced to reconstruct the distribution of the joint random variable up to arbitrary precision (see below for a rigorous statement).
We first outline the algorithm that samples from the waiting time in the \mbox{\genswaponly} protocol following sec.~\ref{sec:Tn}, after which we show how to extend it to track the Werner parameter (sec.~\ref{sec:TnWn}), how to include the communication time after a swap (sec.~\ref{sec:comm-time}) and how to adjust it for the \mbox{\ddistillation} protocol (sec.~\ref{sec:random-variable-extended}).
Pseudocode can be found in algorithm~\ref{alg:monte-carlo}.

We start by explaining the Monte Carlo algorithm for the waiting time in the \mbox{\genswaponly} protocol.
Let $s(X)$ denote a randomized function that yields a sample from the random variable $X$.
We remark that if the cumulative distribution function of $X$ is known, then sampling from $X$ can be done efficiently using inverse transform sampling, which is a standard technique to produce a sample from an arbitrary distribution by evaluating its inverse cumulative distribution function on a sample from the uniform distribution on the interval $[0, 1]$.
We can thus construct the sampler from the waiting time for elementary entanglement, $T_0$, using the inverse of the cumulative distribution function of $T_0$ as given in eq.~\eqref{eq:T0}:
\begin{equation}
	s(T_0) = \lceil \log_{(1-p)}(1 - s(U)) \rceil
	\label{eq:sample-T0}
\end{equation}
	where $U$ is a random variable which is distributed uniformly at random on $[0, 1]$ and $\lceil . \rceil$ denotes the ceiling function.

For sampling from higher levels, we first note that we can easily transform a sampler $s(X)$ into a sampler $s_{\textnormal{sum}} (X, p)$ from a geometric sum $\sum_{j=1}^K X^{(j)}$, where $K$ is geometrically distributed with parameter $p$.
The sampler from the geometric sum probabilistically calls itself:
\begin{align}
	s_{\textnormal{sum}} (X, p) :=
	\begin{cases}
		s(X) & \text{with prob. $p$},\\
		s(X) + s_{\textnormal{sum}} (X, p) & \text{with prob. $1 - p$}.
  	\end{cases} \nonumber
\end{align}
From the recursive expression for the waiting time $T_n$ in sec.~\ref{sec:Tn} it now follows directly that we can construct a sampler from $T_n$ for $n\geq 1$:
\begin{eqnarray*}
	s(T_{n}) &=& s_{\textnormal{sum}} (M_{n}, \pswap)
\end{eqnarray*}
which, per definition of $s_{\textnormal{sum}}$, makes a call to $s(M_n)$ which is given by
\begin{eqnarray*}
	\label{eq:sampler-M-ell}
	s(M_{n}) &=& g_{\textnormal{T}}(s(T_{n-1}),s(T_{n-1}))
\end{eqnarray*}
where $g_{\textnormal{T}}$ is defined in eq.~\eqref{eq:gT}.

Using the Dvoretzky-Kiefer-Wolfowitz inequality \cite{dvoretzky1956asymptotic}, we determine how many samples from $(T_n, W_n)$ we need in order to obtain bounds on its cumulative probabilities.
It follows from this inequality that if $q(t) := \Pr(T_{n} \leq t)$ denotes the cumulative probability function of the waiting time $T_n$ and $q_m(t)$ the empirical cumulative probabilities after having drawn $m$ samples, then the difference between $q$ and $q_m$ is bounded as
\[
	\Pr( |q(t) - q_m(t)| > \epsilon) \leq 2 e^{-2m\epsilon^2}
	\]
for all $t \geq 0$. Thus we can bound the probability that the empirical estimate $q_m(t)$ deviates from $q(t)$ at most $\epsilon$ for any value of $t$ by \mbox{$z = 2 e^{-2m\epsilon^2}$} if the number of samples to draw equals
\begin{equation}
m = -\log(z/2) / (2 \epsilon^2)
\label{eq:number-of-samples}
\end{equation}
Let us emphasize that this number of samples is independent of any parameters of the repeater chain, for instance the number of segments, and thus its contribution to the runtime or space usage of the Monte Carlo algorithm is at most a multiplicative constant, independent of any such parameters.

Following sec.~\ref{sec:TnWn}, we modify the Monte Carlo algorithm to also compute the Werner parameter of the sampled produced entangled pair (for pseudocode see algorithm~\ref{alg:monte-carlo}).
First note that the notation $s(X)$ which samples from a random variable $X$ can also be applied to a joint random variable $(X, Y)$, so that $s((X, Y))$ returns a tuple.
We will now define a sampler $s((T_n, W_n))$ where $(T_n, W_n)$ is the joint random variable representing waiting time and Werner parameter of a $2^n$-segment \mbox{\genswaponly} repeater chain (see sec.~\ref{sec:TnWn}).
For this, we first need to adapt the sampler of the geometric compound sum $s_{\textnormal{sum}}$ to a sampler of the geometric compound forgetting sum (eq.~\eqref{eq:forgetting-sum}) by defining $\hat{s}_{\textnormal{sum}} ((X, Y), p)$ where $X$ and $Y$ are arbitrary random variables and $p\in[0, 1]$ is the parameter of the geometric distribution:
\begin{align}
	\hat{s}_{\textnormal{sum}} ((X, Y), p) :=
	\begin{cases}
		s((X, Y)) \qquad \text{with prob. $p$},\\
		\\
		\pi(s((X,Y))) + \hat{s}_{\textnormal{sum}} ((X, Y), p) \\
		\qquad
		\qquad
		\qquad
		\text{with prob. $1 - p$}.
  	\end{cases} \nonumber
\end{align}
where `+' denotes pairwise addition and $\pi$ is the projector onto the first element of a tuple: $\pi((x,y)) = (x, 0)$ for any numbers $x, y$.

A recursive definition of the joint sampling function from $(T_n, W_n)$ follows directly from the joint expression for waiting time $T_n$ and Werner parameter $W_n$ in
eqs.~\eqref{eq:TnWn}-\eqref{eq:gW}:
\begin{eqnarray}
	\nonumber
	s((T_0, W_0))
	&=& (s(T_0), \werner_0)
	\\
	s((T_{n}, W_{n}))
	&=&
	\hat{s}_{\textnormal{sum}} ((M_{n}, V_{n}), \pswap)
	\label{eq:ssumMnVn}
	\nonumber
	\\
	s ((M_{n}, V_{n}))
	&=&
	g( s(T_{n - 1}, W_{n - 1}), s(T_{n - 1}, W_{n - 1}))
	\label{eq:monte-carlo-g-implementation}
\end{eqnarray}
where $w_0$ is the Werner parameter of each single-hop link at the time it is produced (see sec.~\ref{sec:model}) and the function $g$ is defined in eq.~\eqref{eq:g}.
In this pseudocode for this Monte Carlo algorithm in algorithm~\ref{alg:monte-carlo}, the sampler $s(T_n, W_n)$ is denoted by $\mathrm{sample\_swap}$.

Since the expression for $(T_n, W_n)$ from sec.~\ref{sec:TnWn} assumes that the communication time for the heralding signal after the entanglement swap takes negligible time, it is not included in the Monte Carlo algorithm above.
Fortunately, the adaptation to include this communication time as in sec.~\ref{sec:comm-time} directly carries over to the Monte Carlo algorithm by replacing $g$ in eq.~\eqref{eq:monte-carlo-g-implementation} with
\[
	g^{n}((t_A, w_A), (t_B, w_B)) := (g_{\textnormal{T}}^{n}(t_A, t_B), g_{\textnormal{W}}^{n}((t_A, w_A), (t_B, w_B))
\]
where $g_{\textnormal{T}}^{n}$ and $g_{\textnormal{W}}^{n}$ are defined in eqs.~\eqref{eq:gT-with-communication-time} and \eqref{eq:gW-with-communication-time}.

The time complexity of the Monte Carlo algorithm is a random variable since it is a randomized algorithm.
Every call to $s(T_{n + 1}, W_{n + 1})$ performs the auxiliary function $\hat{s}_{\textnormal{sum}}$ on average $1/\pswap$ times, each of which calls $s(M_{n}, W_{n})$ precisely once and thus $s(T_{n}, W_{n})$ exactly twice by eq.~\eqref{eq:monte-carlo-g-implementation}.
Given access to a constant-time sampler from the uniform distribution on $[0, 1]$, a sample from the base level $s(T_0, W_0)$ can be obtained in constant time, so a simple inductive argument shows that a drawing a single sample from $(T_{n}, W_{n})$ has average runtime $\mathcal{O}\left((2 / \pswap)^{n}\right)$, which equals
\[
	\mathcal{O}\left(
	N^{\log_2 (2 / \pswap)}
	\right)
	\]
which is polynomial in the number of segments $N=2^n$.

Following sec.~\ref{sec:random-variable-extended}, we also adjust the Monte Carlo algorithm to determine the waiting time and average Werner parameter in the \mbox{\ddistillation} repeater chain protocol.
We add a recursive function $\mathrm{sample\_dist}$ in algorithm~\ref{alg:monte-carlo} for sampling from the random variable $T_n^{d'}$ from eq.~\eqref{eq:Tnd}, which represents the waiting time at each level $d'\in \{0, 1, \dots, d\}$ of the nested distillation scheme.
The relation between the random variable tuples $\left(T_n^{d'}, W_n^{d'}\right)$ and $\left(M_n^{d'}, V_n^{d'}\right)$ on the one hand and $\left(T_n, W_n\right)$ and $\left(M_n, V_n\right)$ on the other is mirrored in their implementations $\mathrm{sample\_dist}$ and $\mathrm{sample\_swap}$, respectively: the function $\mathrm{sample\_swap}$ calls $\mathrm{sample\_dist}$ following eq.~\eqref{eq:Mn-dist}, which subsequently calls itself recursively for $d$ nesting levels following eq.~\eqref{eq:Mnd} and eq.~\eqref{eq:Tnd} and calls $\mathrm{sample\_swap}$ at the lowest level in line with eq.~\eqref{eq:Tn0}.
See algorithm~\ref{alg:monte-carlo} for the full pseudocode of the $d$-dependent sampler $\mathrm{sample\_swap}$ for the \mbox{\ddistillation} protocol.

The average runtime of the sampler for the \mbox{\ddistillation} protocol is upper bounded by $\mathcal{O}\left(4^{d} \cdot (2/\pswap)^n\right)$.
In order to derive this, note that the probability that a distillation attempt succeeds (see eq.~\eqref{eq:pdist}) is lower bounded by $1/2$ and hence a call to $\mathrm{sample\_dist}(n, d)$ recursively performs at most $(2 / (1/2))^d = 4^d$ calls to $\mathrm{sample\_swap}(n-1)$ on average.
The average runtime of the full algorithm is the product of this number of calls and the average runtime of the \mbox{\genswaponly} algorithm $\mathcal{O}\left((2/\pswap)^n\right)$ since the recurrence distillation scheme is performed at every swapping level.

Let us finish this section with an analysis of the algorithm's space complexity.
For generating a single sample of $(T_n, W_n)$ of the $\genswaponly$ protocol, the number of variables that need to be stored grows linearly in the number of segments $n$.
To see this, first note that at level $\ell$ the algorithm only needs to keep track of two samples of $(T_{\ell - 1}, W_{\ell - 1})$ at a time, since in the case of a failed swap it may discard the samples after updating the total time used and subsequently reuse the space for storing two fresh samples.
In addition, for producing these two samples, only two samples need to be stored at \textit{every} level $< \ell$.
The insight here is that at each level the required two samples can be drawn \textit{in sequence} rather than in parallel\footnote{Note that in our runtime analysis, we already assumed sequentiality since we showed that the average \textit{number of calls} to $s(T_0, W_0)$ is at most polynomial in the number of segments $n$.}, so that the space needed to draw the first sample can be reused for the second.
Therefore, the algorithm needs to keep track of at most two samples at every level, which implies that the total number of variables it stores is linear in the number of levels and thus in the number of segments $n$.
For the $\ddistillation$ protocol, the scaling is linear in $n \cdot d$ with $d$ the number of distillation steps per nesting level, which can be shown by an analogous argument.

The number of samples that is required to generate a probability distribution histogram with pre-specified precision is independent of the number of segments (see explanation directly below eq.~\eqref{eq:number-of-samples}).
For constructing the histogram, we only need to store the waiting times for which at least a single sample was drawn and hence the number of such waiting times is also independent of the number of segments.
We conclude that reproducing the probability distribution of $(T_n, W_n)$ using the Monte Carlo algorithm will not exceed polynomial space usage in the number of segments $n$.

\begin{algorithm}[ht]
	\LinesNumbered
	\caption{Monte-Carlo algorithm $\mathrm{sample\_swap}(n)$ for producing a single sample of the joint waiting time and Werner parameter $(T_n, W_n)$ for a \mbox{\ddistillation} quantum repeater chain of $2^n$ segments as in sec.~\ref{sec:TnWn}.
		 Setting $d=0$ corresponds to the \mbox{\genswaponly} repeater chain protocol.
		 }
	\label{alg:monte-carlo}
	\SetKwInOut{KwInput}{Input}
    \SetKwInOut{KwOutput}{Output}
    \SetKwComment{Comment}{<start>}{<end>}
	\SetAlgoLined
	\KwInput{Success probabilities $\pgen$ and $\pswap$, Werner parameter of single-hop links $\werner_0$, nesting level $n$, number of distillation rounds at each level $d$.}
	\KwOutput{Single sample from $(T_n, W_n)$.}
	\uIf{n = 0}{
		$u \gets$ uniform random sample from $[0, 1]$ \\
		\Return{$(\lceil \log_{(1-p)}(1 - u) \rceil, \werner_0)$
		\tcp*[f]{(eq.~\eqref{eq:sample-T0})}
		\label{line:monte-carlo-base-level}
		}
	}
	\ElseIf{$n \geq 1$}{
		($t_A, \werner_A) \gets \mathrm{sample\_dist}(n,d)$
		\label{line:mc-first-sample}
		\\
		($t_B, \werner_B) \gets \mathrm{sample\_dist}(n,d)$
		\label{line:mc-second-sample}
		\\
		$t, \werner \gets g((t_A, \werner_A), (t_B, \werner_B))$
		\tcp*[f]{(eq.~\eqref{eq:g})}
		\label{line:twait}
		\\
		$u \gets$ uniform random sample from $[0,1]$ \\
		\uIf{$u \leq \pswap$}{
			\Return{$t, \werner$
			\label{line:mc-successful-swap}
			}
		}
		\Else{
			$t_{\mathrm{retry}}, \werner_{\mathrm{retry}} \gets \mathrm{sample\_swap}(n)$ \\
			\Return{$t + t_{\mathrm{retry}}, \werner_{\mathrm{retry}}$
			\label{line:mc-failed-swap}
			}
		}
	}
	\BlankLine
	\BlankLine
	\SetKwProg{auxiliary}{}
	\auxiliary{}{\textbf{Auxiliary function} $\mathrm{sample\_dist}(n, d):$}\\
	\uIf{d = 0}{
		\Return{$\mathrm{sample\_swap}(n-1)$}
	}
	\Else{
		($t_A, \werner_A) \gets \mathrm{sample\_dist}(n, d-1)$ \\
		($t_B, \werner_B) \gets \mathrm{sample\_dist}(n, d-1)$ \\
		$t, \werner \gets g_{\textnormal{D}} \left((t_A, \werner_A), (t_B, \werner_B)\right)$
		\tcp*[f]{eq.~\eqref{eq:gD}}
		\\
		$u \gets$ uniform random sample from $[0,1]$ \\
		\tcp{Success probability: eq.~\eqref{eq:pdist}}
		\uIf{$u \leq
		p_{\textnormal{dist}} (\werner_A, \werner_B)$
		}{
			\Return{$t, \werner$}
		}
		\Else{
			$t_{\mathrm{retry}}, \werner_{\mathrm{retry}} \gets \mathrm{sample\_dist}(n,d)$ \\
			\Return{$t + t_{\mathrm{retry}}, \werner_{\mathrm{retry}}$}
		}
	}
\end{algorithm}

\begin{algorithm}
	\LinesNumbered
	\caption{Deterministic algorithm for computing the probability distribution of the waiting time $T_n$ of the \mbox{\genswaponly} protocol at nesting level $n$.
	The subroutine $\mathrm{fast\_convolution\_algorithm}$ computes the distribution of the sum of two random variables $A$ and $B$, each represented by an array of size $\ttrunc + 1$ with their probabilities $\Pr(A = t)$ and $\Pr(B = t)$ for $t\in \{0, 1, 2, \dots, \ttrunc\}$.
		 }
	\label{alg:deterministic}
	\SetKwInOut{KwInput}{Input}
    \SetKwInOut{KwOutput}{Output}
    \SetKwComment{Comment}{<start>}{<end>}
	\SetAlgoLined
	\KwInput{Success probs. $\pgen$ and $\pswap$, nesting level $n$}
	\KwOutput{
		Two-dimensional array of size $(n + 1) \times (\ttrunc + 1) $ with entries $\Pr(T_{\ell} = t)$ for $\ell \in \{0, 1, 2, \dots, n\}$ and $t\in \{0, 1, 2, \dots, \ttrunc\}$.
		}
		$C \gets \textnormal{3-dim. array of zeros,}$\\
		$\qquad \textnormal{size $(n + 1) \times (\ttrunc + 1) \times (\ttrunc + 1)$}$\\
		$T \gets \textnormal{2-dim. array of zeros, size $(n + 1) \times (\ttrunc + 1)$}$\\
		$M \gets \textnormal{1-dim. array of zeros, of size $(\ttrunc + 1)$}$\\
		\label{line:M}

		\BlankLine
		\tcp{Base level probs (eq.~\eqref{eq:T0})}
		\For{$t \in \{0, 1, \dots, \ttrunc\}$}{
			$T[0, t] \gets 1 - (1 - \pgen)^t$
		}

		\BlankLine
		\tcp{Probabilities on higher levels}
		\For{$\ell \in \{0, 1, \dots, n - 1\}$}{
			\BlankLine
			\tcp{Maximum of two copies (eq.~\eqref{eq:max-step-computation})}
			\For{$t \in \{1, 2, \dots, \ttrunc\}$}{
				$M[t] \gets T[\ell, t]^2 - T[\ell, t - 1]^2$
			}
			\BlankLine
			\tcp{Conditional probs... (eq.\eqref{eq:convolution})}
			\For{$k \in \{1, 2, \dots, \ttrunc\}$}{
				set column $C[\ell, k]$ to output of $\mathrm{convolve(C[\ell, k-1], k, M)}$
			}
			\tcp{...and the marginals (eq.~\eqref{eq:deterministic-inductive-step})}
			\For{$t \in \{1, 2, \dots, \ttrunc\}$}{
				\For{$k \in \{1, 2, \dots, \ttrunc\}$}{
					$\mathrm{term} \gets \pswap (1 - \pswap)^{k-1}	\cdot C[\ell, k, t]$\\
				add $\mathrm{term}$ to $T[\ell + 1, t]$
				}
			}
			\BlankLine
			\tcp{Convert $T$ to cumulative probs}
			\For{$t \in \{1, 2, \dots, \ttrunc\}$}{
				$T[\ell + 1, t] \gets T[\ell + 1, t] + T[\ell + 1, t -1]$
			}
				\label{line:last-line-of-ell-loop}
		}
		\Return{T}
	\BlankLine
	\BlankLine
	\SetKwProg{auxiliary}{}
	\auxiliary{}{\textbf{Auxiliary function} $\mathrm{convolve}(S, k, M):$}\\
	\uIf{k = 1}{
		\Return{$M$}
	}
	\Else{
		\tcp{Compute convolution of two arrays using Fast Fourier Transforms}
		$\mathrm{array\_with\_sum\_distribution} \gets \mathrm{fast\_convolution\_algorithm}(S, M)$\\
		\Return{$\mathrm{array\_with\_sum\_distribution}$}
	}
\end{algorithm}

\begin{algorithm}
	\SetAlgoLined
	\caption{Extension to algorithm~\ref{alg:deterministic} for computing $W_n$, the average Werner parameter of the end-to-end state produced at time $t \in \{0, 1, \dots, \ttrunc\}$ by a $2^n$-segment repeater chain with the \mbox{\genswaponly} protocol.
	The algorithm is an extension to algorithm~\ref{alg:deterministic} and contains of several parts that should be inserted into that algorithm.
	}
	\label{alg:calculate-fidelity}
	\SetKwInOut{KwInput}{Input}
	\SetKwInOut{KwOutput}{Output}
	\BlankLine
	\text{The following should be inserted directly after line~\ref{line:M}} \\
	\qquad \text{in algorithm~\ref{alg:deterministic}}:\\
	$W \gets$ \text{2-dim. array of zeros, size $(n + 1) \times (\ttrunc + 1)$}\\
	\For{$t \in \{0, 1, 2, \dots, \ttrunc\}$}{
		$W[0, t] \gets w_0$
		}
	\BlankLine
	\text{The following should be inserted directly after line~\ref{line:last-line-of-ell-loop}} \\
	\qquad \text{within the loop over $\ell$ in algorithm~\ref{alg:deterministic}}:\\
	\For{$k \in \{1, 2, \dots,  \ttrunc\}$}{
		\For{$t_A \in \{1, 2, \dots, \ttrunc\}$}{
			\For{$t_B \in \{1, 2, \dots, \ttrunc\}$}{
				\tcp{The average Werner parameter (eq.~\eqref{eq:werner-run})}
					$w_A \gets W[\ell, t_A]$ \\
					$w_B \gets W[\ell, t_B]$ \\
					$t, w \gets g((t_A, w_A), (t_B, w_B))$
					\tcp*[f]{(see eq.~\eqref{eq:g})}
					\\
					\tcp{Add terms to the numerator of eq.~\eqref{eq:Werner-parameter-tracking}...}
					$p \gets \pswap \left( 1- \pswap\right)^{k - 1}$\\
					$p \gets p\cdot T[\ell, t_A] \cdot T[\ell, t_B]$ \\
					\uIf{$k = 1$}{
						\tcp{... in the case of a single swap (eq.~\eqref{eq:prob-single-swap})}
						$W[\ell + 1, t] \gets W[\ell + 1, t] + w \cdot p$\\
						}
					\Else{
						\tcp{... in the case of multiple swaps (eq.~\eqref{eq:prob-multiple-swaps})}
						\For{$t_{\mathrm{fail}} \in \{0,1, 2, \dots, \ttrunc - 1\}$}{
							$t_{\textnormal{deliver}} \gets t + t_{\mathrm{fail}}$\\
							add
							$w \cdot p \cdot C[\ell + 1, k - 1, t_{\mathrm{fail}}]$
							to
							$W[\ell + 1, t_{\textnormal{deliver}}]$
							\\
						}

						}

			}
		}
	}
	\For{$t \in \{1, 2, \dots, \ttrunc\}$}{
		\tcp{Normalize by dividing by the denominator of eq.~\eqref{eq:Werner-parameter-tracking}}
		$W[\ell + 1, t]$ $\gets$ $W[\ell + 1, t] / T[\ell + 1, t]$
	}
\end{algorithm}

\subsection{Second algorithm: deterministic computation \label{sec:deterministic}}

In this section, we present our full second algorithm, which computes the probability distribution of the waiting time and average Werner parameter up to some pre-specified truncation time $\ttrunc$.
The algorithm applies to the \mbox{\genswaponly} repeater protocol.
In what follows, we first show how to compute the probability distribution of the waiting time $T_n$ of the \mbox{\genswaponly} protocol by recursion (see sec.~\ref{sec:Tn}).
After this, we outline how our algorithm performs a modified version of this computation on the finite domain $\{1, 2, \dots, \ttrunc\}$.
We finish the section by extending its computation to include the average Werner parameter (sec.~\ref{sec:TnWn}).

Let us start by showing how to derive the probability distribution of the waiting time $T_n$ in the \mbox{\genswaponly} protocol.
For a single repeater segment ($n=0$), the waiting time follows the geometric distribution as given in eq.~\eqref{eq:T0}.
For nesting levels $\ell \in \{0, 1, 2, \dots, n\}$, the relation between the probability distributions of $M_{\ell}$ and $T_{\ell}$ follows straightforwardly from eq.~\eqref{eq:M0}:
\begin{equation}
\label{eq:max-step-computation}
	\Pr(M_{\ell} = t) = \Pr(T_{\ell} \leq t)^2 - \Pr(T_{\ell} \leq t - 1)^2
	.
\end{equation}
Now we compute the probability distribution of $T_{\ell + 1}|K_{\ell}$, which is the waiting time conditioned on the number of swaps needed that transform $2^{\ell}$-hop entanglement to the final $2^{\ell + 1}$-hop entanglement:
\begin{eqnarray}
	\nonumber
	\Pr\left(T_{\ell + 1} = t | K_{\ell} = k \right)
	&=&
	\Pr\left( \sum_{j=1}^k M_{\ell}^{(j)} = t\right)
	\\
	&=&
	\left[\Conv\limits_{j=1}^k m_{\ell} \right](t)
	\label{eq:convolution}
\end{eqnarray}
where we have denoted $m_{\ell}(t) :=\Pr(M_{\ell} = t)$ and $*$ denotes convolution of functions (see sec.~\ref{sec:notation}).
The marginal probability distribution of $T_{\ell + 1}$ is calculated from the distribution of the conditional random variable $T_{\ell + 1} | K_{\ell}$ as
	\begin{equation}
		\label{eq:deterministic-inductive-step}
		\Pr(T_{\ell + 1} = t)
		=
		\sum_{k=1}^{\infty}
		\pswap (1 - \pswap) ^ {k-1}
		\Pr\left(T_{\ell + 1} = t | K_{\ell} = k \right)
	\end{equation}
where we used the fact that the number of swaps $K_{\ell}$ is geometrically distributed with parameter $\pswap$.

Our algorithm computes the probability distribution of $T_{n}$ by iterating the procedure in the eqs. \eqref{eq:max-step-computation}, \eqref{eq:convolution} and \eqref{eq:deterministic-inductive-step} over $\ell$ from $0$ to $n - 1$ and is outlined in algorithm~\ref{alg:deterministic}.
Its implementation follows naturally from the equations above except for the following remarks.
First, in the algorithm, the sum in eq.~\eqref{eq:deterministic-inductive-step} is truncated at the pre-specified truncation time $\ttrunc$.
That this truncation yields correct probabilities $\Pr(T_{\ell + 1} = t)$ for all $t \in \{0, 1, \dots, \ttrunc\}$ follows from the fact that $\Pr(T_{\ell + 1} = t | K_{\ell} > t) = 0$ since the generation of entanglement over any number of hops takes at least a single time step.
Second, the convolutions in eq.~\eqref{eq:convolution} can be computed iteratively over $k$ by noting that $\Pr(T_{\ell + 1} = t | K_{\ell} = k + 1)$ equals the convolution of $\Pr(T_{\ell + 1} = t| K_{\ell} = k)$ and $m_{\ell} (t)$.
Moreover, for a single convolution we use a well-known algorithm based on Fast Fourier Transforms \cite{cooley1965algorithm} which we denote by $\mathrm{fast\_convolution\_algorithm}$ in algorithm~\ref{alg:deterministic}.
This subroutine computes the convolution of two arrays of size $\ttrunc$ in time $\Theta(\ttrunc \log \ttrunc)$.

The time complexity of the deterministic algorithm~\ref{alg:deterministic} equals $\Theta(n\cdot \ttrunc^2 \log \ttrunc)$: the iteration over a single level is dominated by the $\Theta(\ttrunc^2 \log \ttrunc)$ runtime of the convolutions in eq.~\eqref{eq:convolution} because eqs. \eqref{eq:max-step-computation} and \eqref{eq:deterministic-inductive-step} are performed in linear time in $\ttrunc$ by looping through an array of $\ttrunc$ elements.
In sec.~\ref{sec:truncation-point}, we give an explicit expression for the truncation time $\ttrunc$ which ensures that $\Pr(T_n \leq \ttrunc) \geq 0.99$.
This expression is polynomial in the number of repeater segments, which implies that algorithm~\ref{alg:deterministic} runs in polynomial time in the number of segments also.

We extend our deterministic algorithm to also compute the average Werner parameter $W_n(t) := \mean{W_n | T_n = t}$ of the end-to-end link produced at time $t$ by a $2^n$-segment \mbox{\genswaponly} repeater chain (see sec.~\ref{sec:TnWn}).
The computation of the average Werner parameter at each level from $0$ to $n$ is performed after completion of the computation of the waiting time probabilities at the same level.

Let us explain the algorithm here (see algorithm~\eqref{alg:calculate-fidelity} for pseudocode).
At the base level the fidelity $W_0(t)$ equals the constant Werner parameter $\werner_0$ as in sec.~\ref{sec:TnWn} for all \mbox{$t \in \{1, 2, \dots, \ttrunc\}$}.
At a higher level, the Werner parameter of a link which is delivered at time $t$ is the output of $g_W$ from eq.~\eqref{eq:gW}, averaged over all possible \textit{realizations} of waiting times $T_n$ which yield $T_n = t$.
In order to precisely define what we mean by `realization', note that the waiting time $T_n$ and average Werner parameter $W_n$ as expressed recursively in sec.~\ref{sec:TnWn} are a function of $K_n$ copies of $(T_{n-1}, W_{n-1})$, the waiting time and Werner parameter at one level lower.
Regarding $(T_n, W_n)$ as a function with $K_n$ and all such copies of $(T_{n-1}, W_{n-1})$ as input, we define a `realization' of $(T_n, W_n)$ as its evaluation on particular instances of these copies.

Using the notion of realization, we obtain the Werner parameter of the $2^{\ell}$-hop link at levels \mbox{$\ell \in \{1, 2, \dots, n\}$}, given that it was produced at time $t$:
\begin{eqnarray}
	W_{\ell} (t)
	=
	\frac{
		\sum\limits_{\substack{\textnormal{$r$} : \\ \textnormal{$r$ delivers link at $t$}}} p_{\ell}(r)
	\cdot
	W_{\ell}^{\textnormal{av}}(r)
	}
	{
		\sum\limits_{\substack{\textnormal{$r$} : \\ \textnormal{$r$ delivers link at $t$}}} p_{\ell}(r)
	}
	\label{eq:Werner-parameter-tracking}
\end{eqnarray}
where $r$ is a realization of $(T_{\ell}, W_{\ell})$ and $W_{\ell}^{\textnormal{av}}(r)$ denotes the average Werner parameter of the $2^{\ell}$-hop that realization $r$ delivers with $p_{\ell}(r)$ its probability of occurrence.

In what follows, we will derive expressions for $p_{\ell}(r)$ and $W_{\ell}^{\textnormal{av}}(r)$.
This will give us an explicit expression for $W_{\ell}(t)$ and it is this expression that our algorithm evaluates.
We distinguish between two cases of realizations for computing $p_{\ell}(r)$.
In the first case, only a single swap (i.e. $K_{\ell} = 1$) is needed to produce the $2^{\ell}$-hop entanglement, i.e. the first swap from level $\ell - 1$ to $\ell$ is successful.
The realizations $r$ that belong to this case can be parameterized by the times $t_A$ and $t_B$ at which the two $2^{\ell - 1}$-hop links are generated.
The total probability of occurrence of these realizations, each of which delivers a $2^{\ell}$-hop link at time $g_{\textnormal{T}}(t_A, t_B)$ (see eq.~\eqref{eq:gT}), is given by
\begin{equation}
	\label{eq:prob-single-swap}
	p_{\ell}(r) =
	\Pr(K_{\ell} = 1) \Pr(T_{\ell} = t_A) \Pr(T_{\ell} = t_B)
\end{equation}
and the average Werner parameter of the produced $2^{\ell}$-hop entangled link is
\begin{equation}
\label{eq:werner-run}
	W_{\ell}^{\textnormal{av}} (r) = g_{\textnormal{W}}((t_A, W_{\ell - 1}(t_A)), (t_B, W_{\ell - 1}(t_B)))
\end{equation}
where $g_W$ is given in eq.~\ref{eq:gW}.

In the second case, at least a single entanglement swap to produce $2^{\ell}$-hop entanglement fails.
Note that the average Werner parameter only depends on the states of the two $2^{\ell - 1}$-hops that are produced as input to the \textit{last} swap since the entanglement inputted into the failing swaps is lost.
In the case of multiple swaps we can therefore group together the realizations for which the following four quantities are identical: the waiting times $t_A$ and $t_B$ for the production of the last two $2^{\ell - 1}$-hop links with in addition the number of swaps $k$ and the time $t_{\textnormal{fail}}$ that these failed swaps need.
The total probability of occurrence of such a group of realizations equals the product of four probabilities,
\begin{eqnarray}
	\nonumber
	p_{\ell}(r)
	&=&
	\Pr(K_{\ell} = k)
	\cdot \Pr(T_{\ell} = t_{\textnormal{fail}} | K_{\ell} = k - 1)
	\\
	&&
	\cdot \Pr(T_{\ell} = t_A)
	\cdot \Pr(T_{\ell} = t_B)
	\label{eq:prob-multiple-swaps}
\end{eqnarray}
while the average Werner parameter $W_{\ell}^{\textnormal{av}}(r)$ of the $2^{\ell}$-hop that is produced by each of these realizations is identical to the first case and is given in eq.~\eqref{eq:werner-run}.
Each realization in this group delivers a $2^{\ell}$-hop link at time $t_{\mathrm{fail}} + g_{\textnormal{T}}(t_A, t_B)$ (see eq.~\eqref{eq:gT}).

Our algorithm loops over each group of realizations, evaluates their probabilities of success in eqs.~\eqref{eq:prob-single-swap} and \eqref{eq:prob-multiple-swaps} and their average Werner parameter in eq.~\eqref{eq:werner-run} and subsequently computes $W_{\ell}(t)$ using eq.~\eqref{eq:Werner-parameter-tracking}.
The domain of the time parameters $t_A, t_B$ and $t_{\mathrm{fail}}$ is bounded from above by $\ttrunc$ since no short-range link that is used to produce a long-range link at time $\leq \ttrunc$ can take longer than $\ttrunc$.
Also, the total number of swaps $K_n$ runs up to $\ttrunc$ since it cannot exceed the time at which the end-to-end link is delivered by the same reasoning as the truncation of the sum in eq.~\eqref{eq:deterministic-inductive-step}, i.e. $\Pr\left(T_{n+1} = t | K_n > t\right) = 0$.
The pseudocode of the deterministic algorithm for computing the average Werner parameter can be found in algorithm~\ref{alg:calculate-fidelity}.

The time complexity of the Werner-parameter algorithm can be inferred directly from algorithm~\ref{alg:calculate-fidelity} by the four loops with domain of size $\Theta(\ttrunc)$, which implies that the full time complexity is $\Theta(n \cdot \ttrunc^4)$. This is polynomial in the number of repeater chain segments (see sec.~\ref{sec:truncation-point}).

\subsection{Possible extensions}
In this section, we give examples of possible extensions of the Monte Carlo algorithm and the deterministic algorithm.
First, we provide an example of how the two algorithms can be extended to different quantum state and noise models than the Werner states and depolarizing decoherence noise used in this work.
We also give an example of an extension to a different network topology than a chain.
We finish the section by sketching what is needed to extend the deterministic algorithm to the \ddistillation\ protocol in the future.

An example of applying the algorithms to more general quantum states is to track states that are diagonal in the Bell basis, i.e.\ we assume that the generated single-hop states can be written as
\[
\sum_{j\in\{\pm\}} \sum_{k \in \{\pm\}} p_{j, k} \dyad{\phi_{j, k}}
\]
where \mbox{$\ket{\phi_{+ \pm}} := (\ket{0}\otimes \ket{0} \pm \ket{1}\otimes\ket{1}) / \sqrt{2}$} and \mbox{$\ket{\phi_{- \pm}} := (\ket{0}\otimes\ket{1} \pm \ket{1}\otimes\ket{0}) / \sqrt{2}$} are the four Bell states and the Bell coefficients $p_{j, k}$ are probabilities which sum to 1.
The implementation of the Monte Carlo method would have the Werner parameter $W_n$ replaced by a joint random variable on the four\footnote{In fact, tracking only three of these coefficients already completely characterize a Bell-diagonal state since they sum to one.} Bell coefficients $(p_{++}, p_{+-}, p_{-+}, p_{--})$, while the deterministic algorithm would compute the average over each of these four coefficients individually in a fashion similar to the average of the Werner parameter (eq.~\eqref{eq:Werner-parameter-tracking}).
Successful entanglement swap and distillation operations both map Bell-diagonal states to Bell-diagonal states~\cite{duer2007entanglement} and could thus each be formulated as an operation on the four Bell coefficients.
\\
An example of a different model of memory decoherence noise (currently eq.~\eqref{eq:decay}) is the application of the Pauli operator \mbox{$Z:= \dyad{0} - \dyad{1}$} on one of the two qubits with probability 
\[
q(\Delta t) := \frac{1}{2}\left(1 + e^{-\Delta t / T_{\text{coh}}}\right)
\]
where $\Delta t$ is the time that the state has resided in memory and $T_{\textnormal{coh}}$ is the joint coherence time of the two memories that hold the two qubits.
This probabilistic application of $Z$ acts on the Bell coefficients as
\[
p_{j, k} \mapsto 
q(\Delta t)
\cdot p_{j, k}
+
(1 - q(\Delta t))
\cdot p_{m, \ell}
\]
where $p_{m, \ell}$ is the coefficient belonging to \mbox{$\ket{\phi_{m, \ell}} := \left(\unit_2 \otimes Z\right) \ket{\phi_{j,k}}$} with \mbox{$\unit_2 := \dyad{0} + \dyad{1}$}.
Lastly, the algorithm could be generalized by modelling the swapping and distillation operations as noisy operations by concatenating the perfect operation with a noise map that can be written as operation on the four Bell coefficients.

In addition to more general state and noise models, both algorithms could also be applied to more general network topologies than a chain.
An example is the generation of a Greenberger-Horne-Zeilinger (GHZ) state~\cite{greenberger1989beyond} in a star network, where there is a single central node and each of the other nodes (the leaves) is connected to this single central node only.
All leave nodes start by generating an elementary link with the central node, after which the central node performs a local operation to convert these links into a single GHZ state on all the leave nodes, e.g.\ by a combination of two-qubit controlled-rotation gates and single-qubit measurements~\cite{cirac1999distributed}.
Similar to our model of the swap operation, we could model the local operation that produces the GHZ state as probabilistic, motivated by probabilistic two-qubit operations in linear photonics~\cite{nielsen2004optical}.
In the same spirit as the $\genswaponly$ protocol and our analysis of it in sec.~\ref{sec:random-variable}, the central node waits for all links to have been generated (which corresponds to the maximum of their individual waiting times) while failure of the local operation requires regeneration of the elementary links (which corresponds to the geometric compound sum).
Since both maximums and geometric sums of random variables can be treated by the two algorithms, both could be used to sample the produced state and waiting time in the star network.

In general, the Monte Carlo algorithm could be applied if one can express the random variables describing waiting time and fidelity as a function of random variables whose probability distribution is known.
For the deterministic algorithm, on the other hand, one should find efficiently computable expressions for the probability distributions of these random variables.
These two cases are different: indeed, for the \ddistillation\ protocol, we have found a Monte Carlo algorithm but have been unable to formulate a deterministic algorithm.
We finish this section by elaborating on this difference for the specific case of the \ddistillation\ protocol and sketch an opportunity for extending the deterministic algorithm to it.
For the \genswaponly\ protocol, we were able to design a Monte Carlo sampler of the waiting time and fidelity by expressing the corresponding random variables at level $\ell$ as a function of those at level $\ell - 1$.
In addition, we could express the probability distributions recursively also, which enabled us to formulate the deterministic algorithm.
For the \ddistillation\ protocol, we succeeded in performing the first step, i.e.\ find a recursive expression for the random variables (see sec.~\ref{sec:random-variable-extended}), but have been unable to express the probability distributions in a recursive fashion.
The difficulty lies in the fact that waiting time cannot be analyzed independently of the fidelity of the input states, since the distillation success probability is a function of the produced states contrary to the swapping success probability $\pswap$ (see also last paragraph in sec.~\ref{sec:random-variable-extended} for more explanation).
The resulting two-way dependence between waiting time and fidelity renders eq.~\eqref{eq:deterministic-inductive-step} incorrect if we simply replaced $\pswap$ by the distillation success probability.
For this reason, the key to finding a deterministic algorithm for computing waiting time and fidelity for the \ddistillation\ protocol is to find an efficiently computable expression for the probability function of $(T_{\ell}, W_{\ell})$ as a function of $(M_{\ell - 1}, V_{\ell - 1})$, i.e.\ a correct alternative to eq.~\eqref{eq:deterministic-inductive-step}.

\section{Bounds on the mean waiting time \label{sec:bounds}}
In this section, we first show how to use the deterministic algorithm from section~\ref{sec:deterministic} to obtain bounds on the mean of the waiting time $T_n$, which improve upon a common analytical approximation.
Then we give an explicit expression for the choice of truncation time in the algorithm for which $99\%$ of probability mass of $T_n$ is captured.

\subsection{Numerical mean using the deterministic algorithm \label{sec:numerical-bounds}}
Here, we show how to obtain bounds on the mean of $T_n$ using the deterministic algorithm from section~\ref{sec:deterministic}.
Such bounds are interesting since a common approximation to the mean in the regime of small success probabilities $\pgen$ and $\pswap$, the \textit{3-over-2-formula} \cite{jiang2007fast, simon2007quantum, brask2008memory, sangouard2011quantum}
\begin{equation}
	\label{eq:3-over-2-approximation}
	\mean{T_n} \approx \left(\frac{3}{2\pswap}\right)^n \cdot \frac{1}{\pgen}
	,
\end{equation}
overestimates the waiting time for large success probabilities.
For example, it can be seen in~\cite[fig. 7(a)]{shchukin2017waiting} (reproduced in this work as fig.~\ref{fig:ratio}, top plot) that for $\pgen = \pswap = 1$, the ratio [true mean]/[approximation] of the true mean $\mean{T_n}$ and the approximation in eq.~\eqref{eq:3-over-2-approximation} decreases as a function of the number of segments and equals $0.2$ for a chain of $16$ segments, i.e.\ an overestimation by a factor $\frac{1}{0.2} = 5$.
In fig.~\ref{fig:ratio} (bottom plot), it can be seen that this overestimation grows to more than $\frac{1}{0.05} = 20$ for a chain of $2048$ segments.

The bounds are obtained in two steps.
First, we perform the deterministic algorithm to compute the probability distribution of $T_n$, as described in section~\ref{sec:deterministic}.
Since this probability distribution is only computed by the algorithm on the truncated domain $\{0, 1, \dots, \ttrunc\}$, we cannot calculate the mean of $T_n$.
Instead, we compute its \textit{empirical mean}, which we define for random variable $X$ with the nonnegative integers as domain as
\begin{equation}
\label{eq:computed-mean}
	\meanpartial{X}{\ttrunc} := \sum_{t=1}^{\ttrunc} \Pr(X \geq t)
	.
\end{equation}
Note that the empirical mean reduces to the real mean for $\ttrunc \rightarrow \infty$ (see section~\ref{sec:notation}).

In the second step, we quantify how well the empirical mean of $T_n$ approximates its real mean.
We need two tools for doing so.
As first tool, we introduce the random variable $\Tupper_n$, which is identical to $T_n$ except for the fact that the two links required for the entanglement swap are produced sequentially at every level rather than in parallel.
We proceed analogously to the first step: we perform a modified version of the deterministic algorithm to compute the probability distribution of $\Tupper_n$ and we compute its empirical mean (details and formal definition of $\Tupper_n$ can be found in appendix~\ref{app:bounds}).
In contrast to $T_n$, we are able to compute the real mean of $\Tupper_n$, which equals $E[\Tupper_n] = (2/\pswap)^n \cdot 1/\pgen$ (proof in appendix~\ref{app:bounds}).
The second tool is the following proposition, which states that for $T_n$ the empirical mean converges at least as fast to the real mean with increasing truncation time as for $\Tupper_n$.
\begin{prop}
\label{prop:mean-computation}
	The difference between the real mean and the empirical mean (eq.~\eqref{eq:computed-mean}) of the waiting time is bounded as
\[
	0 \leq \mean{T_n} - \meanpartial{T_n}{\ttrunc} \leq \left(\frac{2}{\pswap}\right)^n \cdot \frac{1}{\pgen} - \meanpartial{\Tupper_n}{\ttrunc}
\]
and the two bounds coincide for $\ttrunc \rightarrow \infty$.
The random variable $\Tupper_n$ is formally defined in appendix~\ref{app:bounds}.
\end{prop}
The main tool for proving proposition~\ref{prop:mean-computation} is the fact that $\Tupper_n$ stochastically dominates $T_n$ for every nesting level $n$, which means that $\Pr(T_n \geq t) \leq \Pr(\Tupper_n \geq t)$ for all \mbox{$t \in \{0, 1, 2, \dots\}$}.
We formally prove the proposition and give a more detailed version of the computation of the probability distribution of $\Tupper_n$ in appendix~\ref{app:bounds}.

\subsection{Choosing a truncation time \label{sec:truncation-point} for the deterministic algorithm}
The truncation time that is inputted into the deterministic algorithm determines how much probability mass will be captured by the algorithm.
The captured probability mass can be bounded from above using Markov's inequality:
	\begin{equation}
		\label{eq:markov}
	\Pr(T_n \geq \ttrunc) \leq \mean{T_n} / \ttrunc
	.
	\end{equation}
We upper bound the mean of $T_n$ in eq.~\eqref{eq:markov} by invoking proposition~\ref{prop:mean-computation} with $\ttrunc = 0$.
The latter reduces to \mbox{$\mean{T_n} \leq (2/\pswap)^n \cdot 1/\pgen$} and thus implies
\[
	\Pr(T_n \geq \ttrunc) \leq
	\left(\frac{2}{\pswap}\right)^n \cdot \frac{1}{\pgen \cdot \ttrunc}
	.
	\]
Consequently, setting
\begin{equation}
	\label{eq:ttrunc-bound}
	\ttrunc =
	\left(\frac{2}{\pswap}\right)^n \cdot \frac{1}{\pgen}
	\cdot \frac{1}{1 - 0.99}
\end{equation}
ensures that an end-to-end link will be produced with probability $\Pr(T_n < t) = 99\%$.

\begin{figure}[h!]
	\centering
	\includegraphics[width=\linewidth]{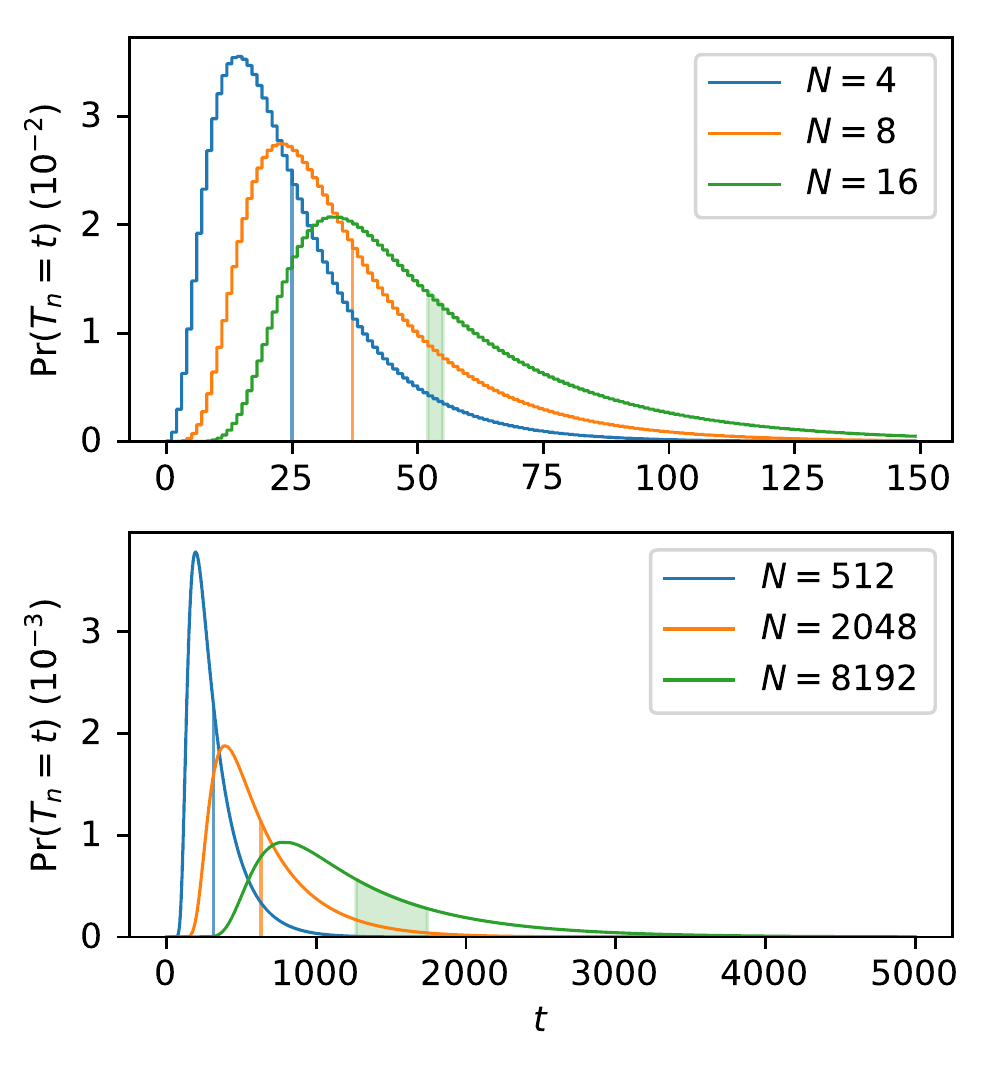}
	\caption{
		Probability distributions of the waiting time and bounds on the mean (vertical shaded areas, see~\ref{sec:numerical-bounds}) calculated by the deterministic algorithm for the \mbox{\genswaponly} protocol.
		The repeater chain parameters are $\pgen=0.1$, $\pswap=0.9$, and the number of repeater segments is given by $N$.
		The top plot recovers the results from Shchukin et al. \cite[Fig.10(a)]{shchukin2017waiting}.
		Computation time $\approx$ 5 seconds for $N=8192$.
		}
	\label{fig:high_prob_swap_inc_shchukin}
	\label{fig:high_prob_swap_pmfs}
\end{figure}

\begin{figure*}[h!]
	\centering
	\includegraphics[width=0.9\textwidth]{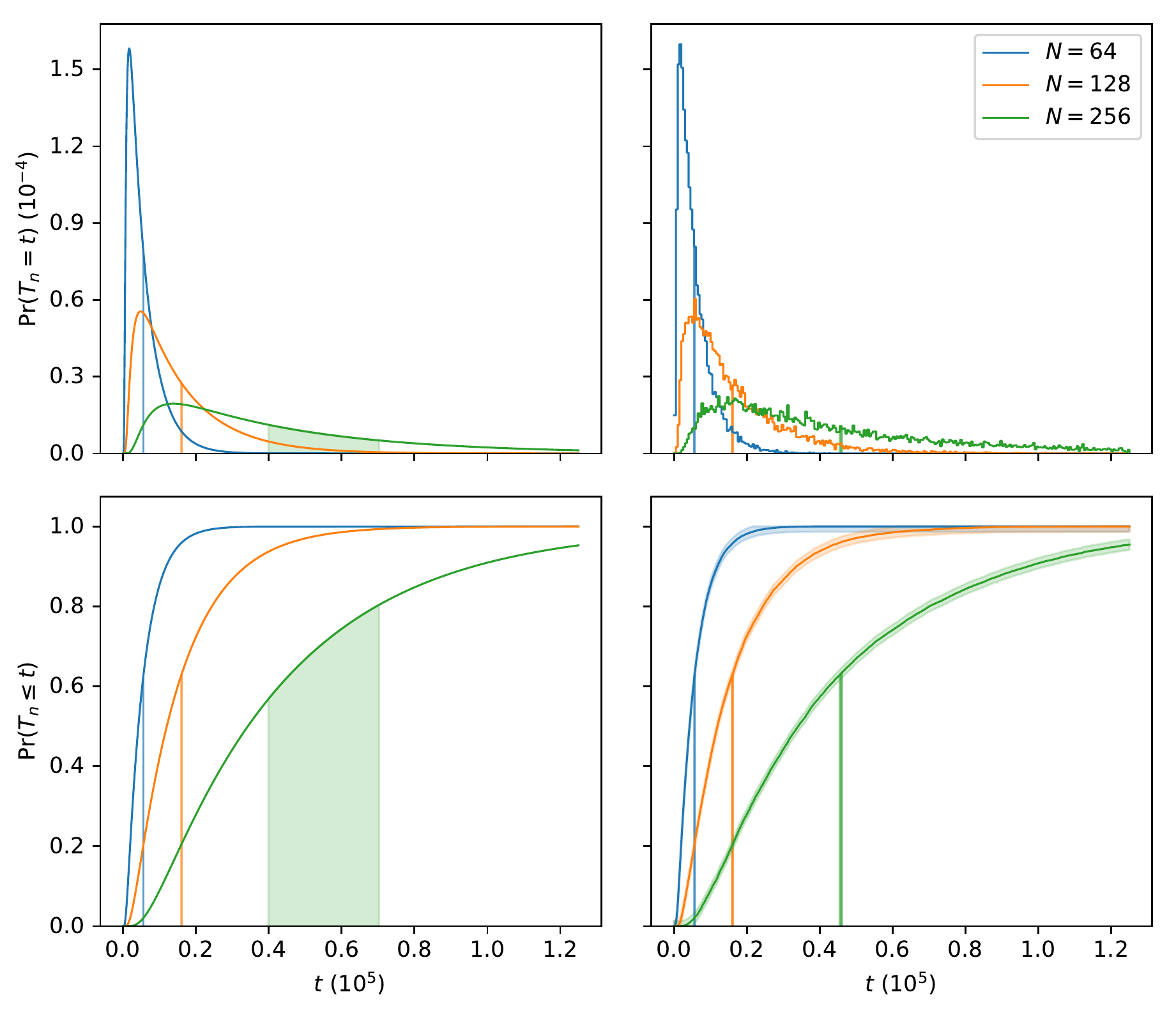}
	\caption{
		The probability distributions of waiting time for repeater chains with varying number of segments $N$ for the \mbox{\genswaponly} protocol, computed using the deterministic algorithm (left plots) and generated from 15,000 samples from the Monte Carlo simulation (right plots). 
		Vertical shaded areas indicate bounds on the mean in the left plots (see sec.~\ref{sec:numerical-bounds}) and the sample mean $\pm$ one standard error in the right plot.
		In the bottom right plot the shaded bands are confidence bands for a probability $z=0.01$ that the actual distribution lies outside of these bands, obtained from the Dvoretzky-Kiefer-Wolfowitz inequality \cite{dvoretzky1956asymptotic} discussed in sec.~\ref{sec:monte-carlo}.
		The repeater chain parameters are $\pgen=0.1$, $\pswap=0.5$.
		The computation time for the deterministic algorithm $\approx$ 9 minutes, while for the Monte Carlo algorithm $\approx$ 30 minutes.
		}
	\label{fig:plots-deterministic-vs-monte-carlo}
\end{figure*} 

\section{Numerical results \label{sec:numerical-results}}
In this section we investigate different repeater chain protocols with the help of our two algorithms.
We start with the \mbox{\genswaponly} protocol, first considering the waiting time distribution of the first produced end-to-end link and subsequently also its average fidelity.
We also show how fidelity and waiting time are affected by the \mbox{\ddistillation} protocol.
Finally we consider the effect of including the communication time in swap operations.

Our proof-of-principle implementation can be found in \cite{githubrepo}.
The reported computation times have been obtained from single-threaded computations on commodity hardware (specifically: a single logical processor of an Intel i7-4770K CPU @ 3.85 GHz).
In the plot captions in this section, we state the computation time for the largest number of repeater segments because computing the distribution of $(T_n, W_n)$ requires finding the distribution of $(T_{n-1}, W_{n-1})$ first (see sec.~\ref{sec:random-variable}).

First we consider the waiting time in the \mbox{\genswaponly} protocol.
Our algorithms are able to recover the results from Shchukin et al. \cite{shchukin2017waiting}, both the full distribution of waiting time exactly (fig.~\ref{fig:high_prob_swap_inc_shchukin}, top plot) as well as its mean up to arbitrary precision (fig.~\ref{fig:ratio}, top plot), and extend these results from 16 to 8192 and to 2048 repeater segments, respectively (figs.~\ref{fig:high_prob_swap_pmfs}, \ref{fig:plots-deterministic-vs-monte-carlo} and \ref{fig:ratio}).
In fig.~\ref{fig:plots-deterministic-vs-monte-carlo} we compare results from both our Monte Carlo and deterministic algorithm and find that there is good agreement between the two.
For high swapping success probability $\pswap$ the deterministic algorithm can compute probability distributions up to thousands of nodes, as illustrated in fig.~\ref{fig:high_prob_swap_pmfs}.
For small $\pswap$, we have found that the number of repeater segments $N = 2^n$ that we can simulate is limited in practice.
This is a consequence of the fact that $\ttrunc$ grows fast in $N$ for small $\pswap$ if we want the guarantee that $99\%$ of the probability distribution is captured (see eq.~\eqref{eq:ttrunc-bound}), and the polynomial scaling in $\ttrunc$ of the algorithm's runtime.

Secondly, we consider the average fidelity of the \mbox{\genswaponly} and \mbox{\ddistillation} protocols.
We investigate the \mbox{\genswaponly} protocol with a small number of segments ($N=1,2,4$), see fig.~\ref{fig:fidelity-plots}.
We observe that fidelity stabilizes as the waiting time increases, and it stabilizes at values for which the state remains entangled in spite of the absence of distillation.
Again, the deterministic and Monte Carlo algorithms show good agreement.
Adding the calculation of fidelity increases the time complexity of the deterministic algorithm, which reduced the maximum number of segments that we could simulate.
We found that the Monte Carlo algorithm is able to simulate a larger number of segments, as its computational complexity is unchanged when also tracking the fidelity.

Third, we consider the \mbox{\ddistillation} protocol.
In fig.~\ref{fig:distillation-and-memory-lifetime}, we study the effects of distillation in a repeater chain of 4 segments comparing one distillation round ($d=1$) against no distillation rounds ($d=0$) for two different memory coherence times.
We first observe the increase in the waiting times caused by the generation of the additional links necessary for distillation.
An increase in waiting time is accompanied by an increase in memory decoherence, which implies that the degree to which distillation is beneficial depends on the memory coherence time.
The values for the coherence time we chose allow to show both types of behavior.

\begin{figure}[h!]
	\centering
	\includegraphics[width=\linewidth]{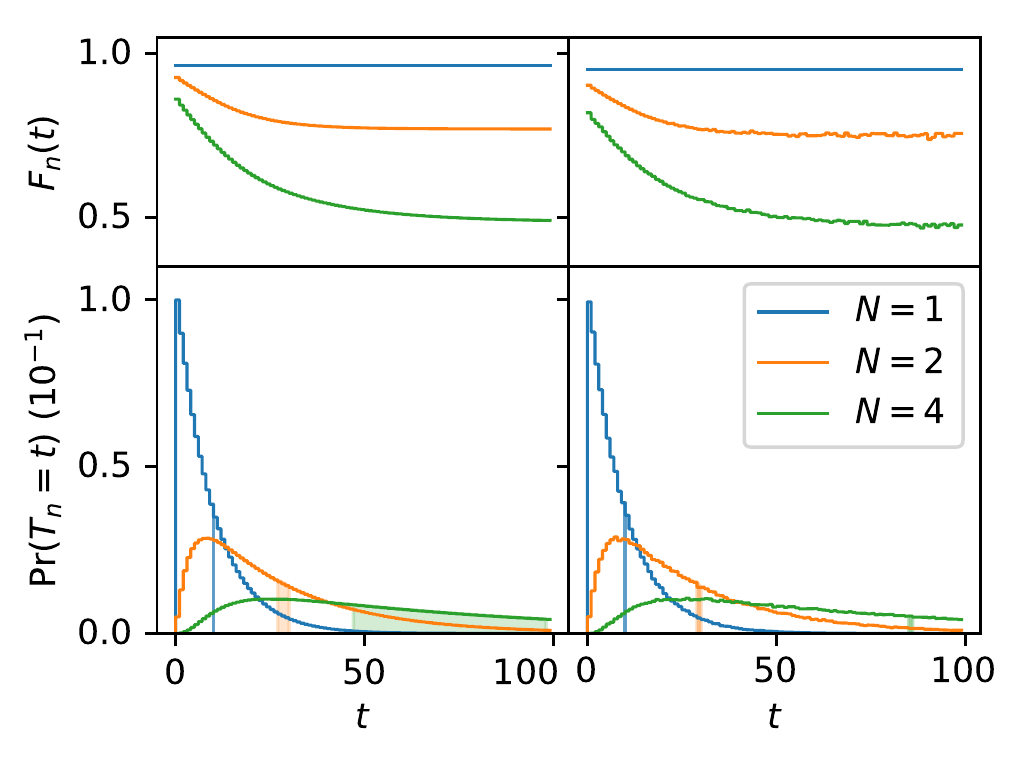}
	\caption{The average fidelity of links delivered at time $t$ by an $N$-segment \mbox{\genswaponly} repeater chain (top row), and the corresponding probability distributions (bottom row), from both the deterministic algorithm (right column) and the Monte Carlo algorithm (left column) using 250,000 samples.
	For the deterministic figures the vertical shaded areas indicate numerical bounds on the mean (see sec.~\ref{sec:numerical-bounds}), and for the Monte Carlo figures these indicate the sample mean $\pm$ one standard error.
	The repeater chain parameters for \mbox{\genswaponly} protocol are $\pgen=0.1$, $\pswap=0.5$, $T_{\textnormal{coh}} = 50$ time steps and the fidelity of the elementary links equals $F_0 = 0.95$, which corresponds to Werner parameter $w_0 = (4\cdot 0.95 - 1)/3 \approx 0.93$ following eq.~\eqref{eq:fidelity-werner-states}.
	The computation time for the deterministic algorithm $\approx$ 15 minutes, while for the Monte Carlo algorithm $\approx$ 20 seconds.
	}
	\label{fig:fidelity-plots}
\end{figure}

\begin{figure}[h!]
	\centering
	\includegraphics[width=\linewidth]{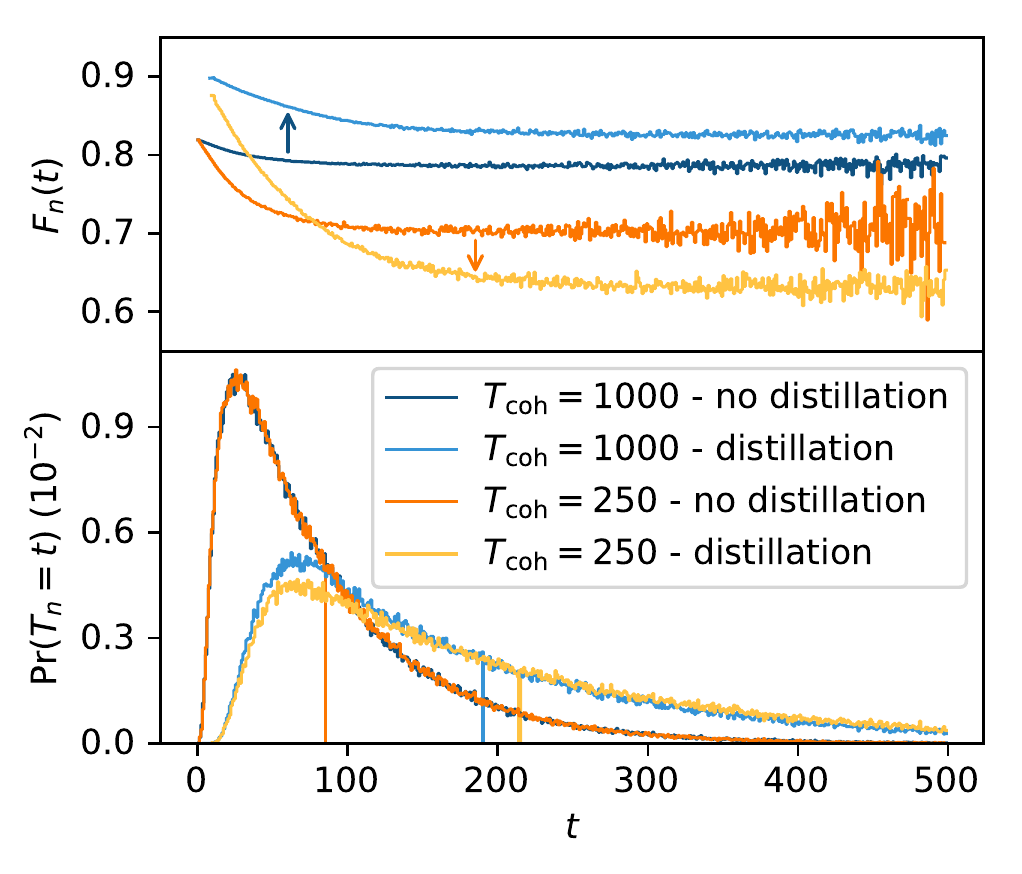}
	\caption{
		Comparison between 
		the \mbox{\ddistillation} protocol with a single distillation round on every level ($d=1$) and the \mbox{\genswaponly} protocol (no distillation) for a repeater chain with $N=4$ segments, for longer and shorter memory coherence times $T_{\text{coh}}$.
		While the goal of distillation is to improve the fidelity of delivered links, when the coherence time is too short compared to the time needed to deliver a link, adding distillation actually decreases the fidelity (orange arrow). For longer coherence times adding distillation does improve the fidelity (blue arrow).
	In both cases the waiting time increases because entanglement distillation requires more links to be generated.
	For the \mbox{\genswaponly} protocol, the waiting time is independent of the memory coherence time (in contrast to fidelity), which can be observed from the identical waiting times in the bottom plot.
	Each curve has been generated from 250,000 Monte Carlo algorithm samples.
	Computation time $\approx$ 15 seconds without distillation, $\approx$ 100 seconds with distillation.
	}
	\label{fig:distillation-and-memory-lifetime}
\end{figure}

Finally, we incorporate the communication time for the entanglement swap into our model following sec.~\ref{sec:monte-carlo}.
Fig.~\ref{fig:swap_comm_time} shows how the output probability distributions change when we include this communication time.
We confirm that, as stated by Brask and S{\o}rensen \cite{brask2008memory}, omitting this communication time gives a good approximation for small $\pgen$, but not for larger $\pgen$.

\begin{figure}[h!]
	\centering
	\includegraphics[width=\linewidth]{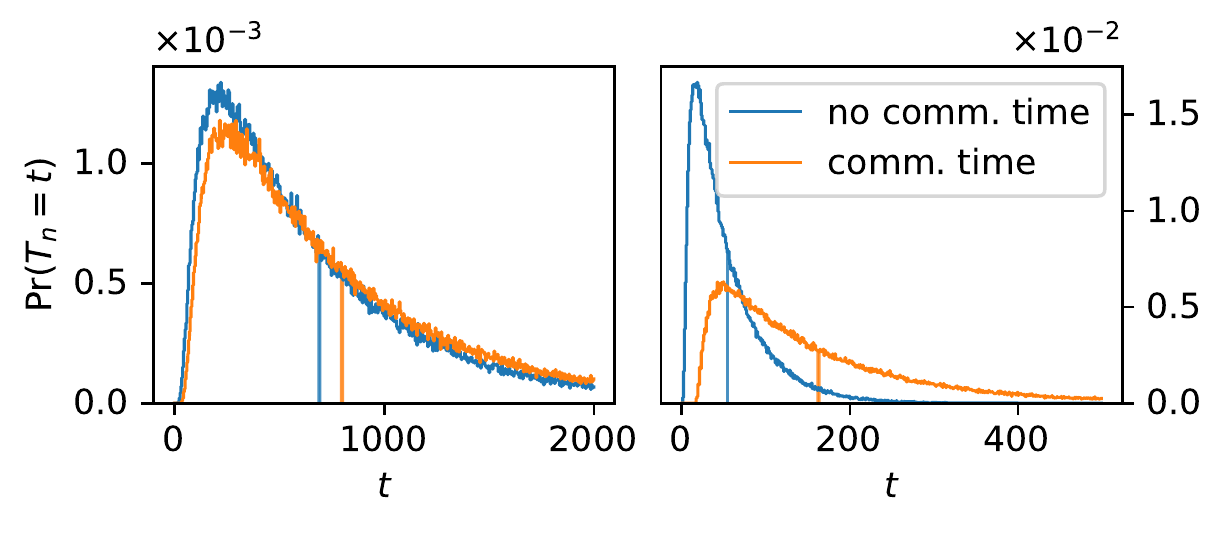}
	\caption{Waiting time distributions with and without communication time for entanglement swapping for the \mbox{\genswaponly} protocol, with entanglement generation success probabilities $\pgen=0.1$ (left) and $\pgen=0.9$ (right), generated from 250,000 samples from the Monte Carlo algorithm. The vertical bars indicate the mean $\pm$ one standard error. As stated by Brask and S{\o}rensen \cite{brask2008memory} omitting this communication time gives a good approximation when $\pgen$ is small.
		The repeater chain has $N=16$ segments and $\pswap=0.5$. Computation time $\approx$ 5 minutes per curve.}
	\label{fig:swap_comm_time}
\end{figure}
\begin{figure}[h!]
\centering
\includegraphics[width=\linewidth]{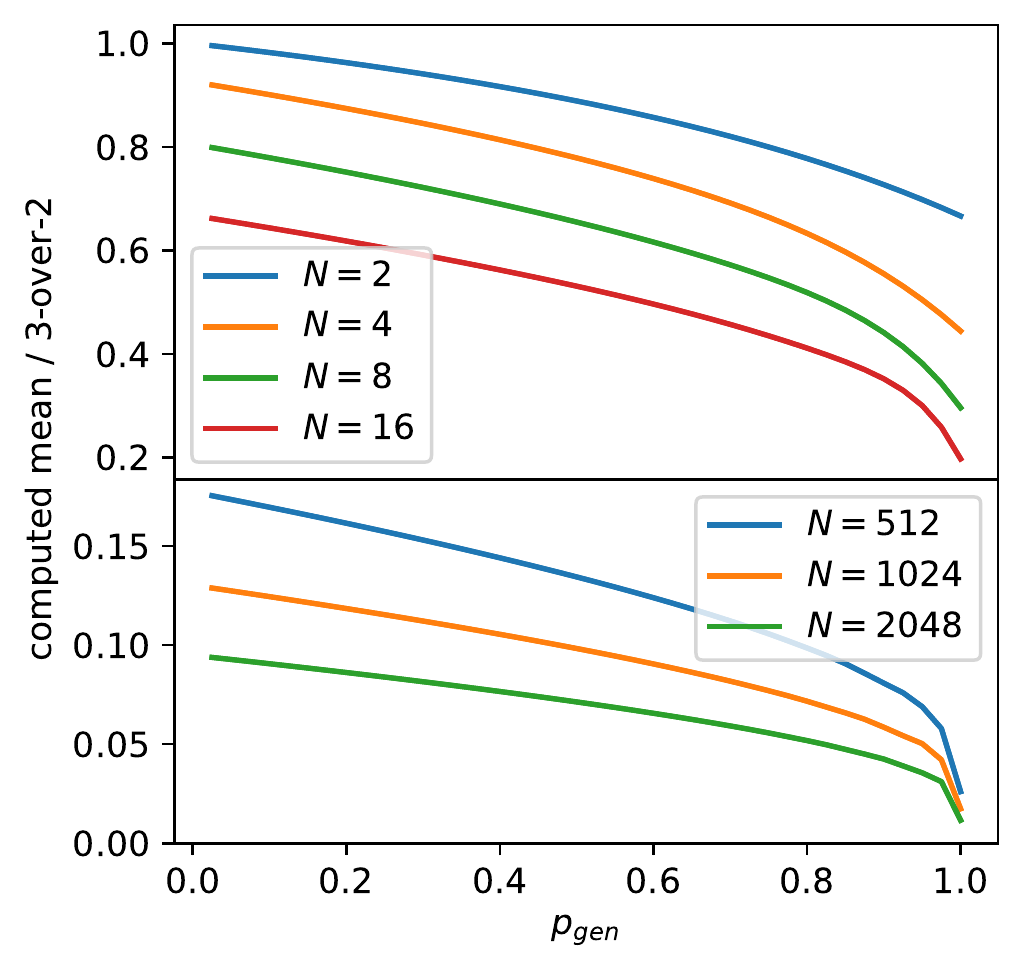}
\caption{\label{fig:ratio}
Ratio between the bound on the mean computed by the deterministic algorithm and the 3-over-2 approximation eq.~\eqref{eq:3-over-2-approximation} as a function of entanglement generation success probability $\pgen$ with deterministic swapping ($\pswap = 1$).
For each number of segments $N$, the figures show two lines: one for the upper and lower bound on the mean (see sec.~\ref{sec:numerical-bounds}).
The fact that for each $N$ only a single thick line rather than two lines can be seen indicates that the bounds on the mean almost coincide. 
The top figure recovers work by Shchukin et al.~\cite[fig. 7(a)]{shchukin2017waiting}, whose exponential-time algorithm based on Markov chains is able to compute the mean exactly while our algorithms can get arbitrarily tight bounds on the mean (deterministic algorithm) or approximate the mean with arbitrary precision (Monte Carlo algorithm) at the benefit of polynomial runtime.
The runtime improvement over the Markov-chain approach allows us to extend the results of Shchukin et al. to more than 2000 segments (bottom figure).
Each curve was generated by running the deterministic algorithm for 40 different values of $\pgen$ ($0.025, 0.05, \dots, 1$) and the truncation time was set to $\ttrunc = 1000$. 
Computation time for each curve $\lesssim$ 2 seconds.
}
\end{figure}

In order to get a rough analytical understanding of the probability distributions for the waiting time that our algorithms have computed, we fitted a generalized extreme value (GEV) distribution to them, which has cumulative distribution function
\begin{equation}
\label{eq:gev-distribution}
\Pr(X \leq t) = \exp(- (1 + \xi s)^{-1/\xi})
\end{equation}
where $X$ is a random variable following the GEV distribution, $s = (t - \mu) / \sigma$, and $\xi > 0$, $\sigma > 0$ and $\mu \in \mathbb{R}$ are the free parameters~\cite{charras2013extreme}.
Fig.~\ref{fig:fit} shows a typical result of such a fit.
We find that the fit seems rather close to the computed distribution, although the difference in the means indicates that the fit should only be used to make approximate statements.

A good fit could provide a speedup for the deterministic algorithm, since the algorithm computes the distribution at each level from the distribution at the previous level.
To be precise, the algorithm's runtime can be reduced by starting the computation at the fitted distribution instead of computing the distribution at level $n$, and subsequently using this distribution to have the algorithm compute the distribution at level $n+1$.
In order to ensure that the distribution at the final level $> n$ still approximates the real distribution, careful analysis of the acquired error of the distribution at higher levels is required.
We leave such error accumulation analysis, based on a distribution that forms a phenomenological fit, for future work.

\begin{figure}[h!]
\centering
\includegraphics[width=\linewidth]{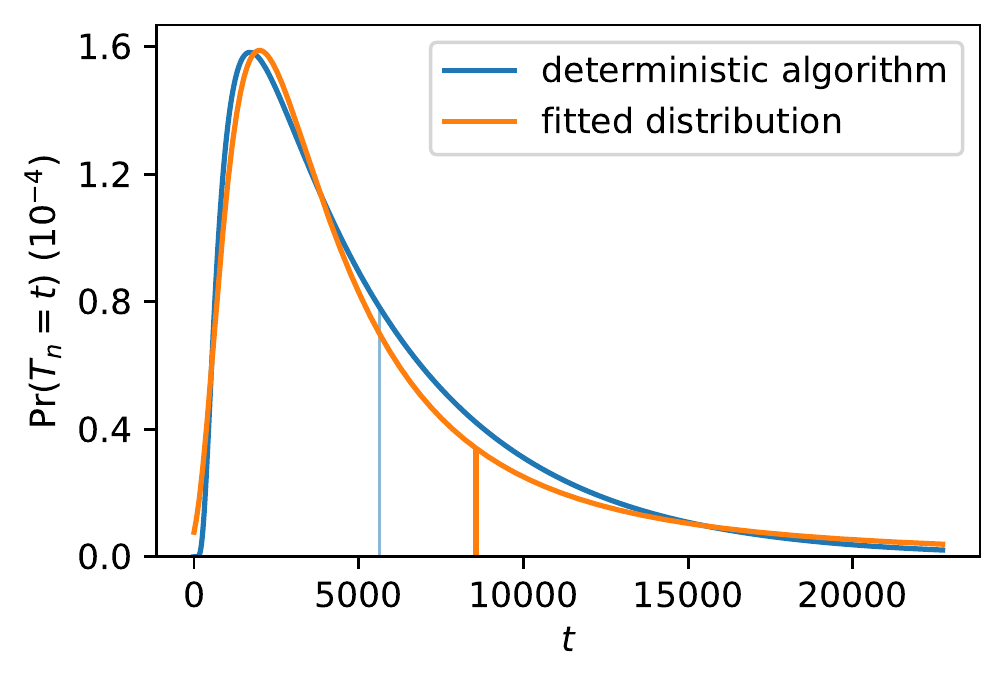}
\caption{\label{fig:fit}
Waiting time distribution (64 segments, $\pgen=0.1$, $\pswap=0.5$), computed with the deterministic algorithm (sec.~\ref{sec:deterministic}), and a fit to the same distribution using the generalized extreme value (GEV) distribution (see eq.~\eqref{eq:gev-distribution}; fitting parameters: $\xi\approx -0.5997, \mu\approx 3092.7, \sigma\approx 2694.9$).
Although the two distributions seem to coincide fairly well by eye, the difference between the means (vertical lines) is relatively large (a factor $\approx 1.5$).
For different number of segments and success probabilities $\pgen$ and $\pswap$, the difference between fitted and computed distribution is similar.
}
\end{figure}

\section{Discussion \label{sec:discussion}}
Quantum networks enable the implementation of communication tasks with qualitative advantages with respect to classical networks.
A key ingredient is the delivery of entanglement between the relevant parties.
In this work, we provide two algorithms for computing the probability that an entangled pair is produced by a quantum repeater chain at any given time and also show how to compute the pair's Bell-state fidelity.
The first algorithm is a probabilistic Monte Carlo algorithm whose precision can be rigorously estimated using standard techniques.
The second one is deterministic and exact up to a chosen truncation time.

Both algorithms run in time polynomial in the number of nodes, which is faster than the exponential runtime of previous algorithms.
The workhorse behind the improved complexity is a formal recursive definition of the waiting time and the state produced by the chain.
We developed an open source proof-of-principle implementation in \cite{githubrepo}, which allows to analyze repeater chains with several thousands of segments for some parameter regimes.

The deterministic algorithm is the fastest of the two for a large set of success probabilities for generating single-hop entanglement $\pgen$ and entanglement swapping $\pswap$.
The Monte Carlo algorithm could be sped up to a factor proportional to the number of samples by parallelization.
A second option to reduce the runtime would be to construct estimates of the random variables at each level of a repeater chain and sample from the estimates to estimate the following level.
A careful analysis would be necessary to ensure that the speed up does not vanish when taking the accumulated precision error into account.

We have been able to adapt our algorithms to several protocols for repeater chains.
More concretely, we have studied the \mbox{\genswaponly} protocol and two generalizations. The first one includes distillation \mbox{\ddistillation}, the second one takes into account the communication time in the swap operations.
We believe that the tools we have developed here can be extended to several other protocols without losing the polynomial runtime.
Some examples which we leave for further work include tracking the full density matrix, variations of \mbox{\ddistillation} with unequal spacing of the nodes or with a number of segments different than a power of two, and the investigation of more general network topologies.
Inspired by hardware, it would also be interesting to model decaying memory efficiency and nodes that can not generate entanglement concurrently with both adjacent neighbors.

In summary, we have proposed two efficient algorithms to characterize the behavior of repeater chain protocols.
We expect our algorithms to find use in the study and analysis of future quantum networks.
Moreover, the existence of protocols capable of efficiently characterizing the state produced opens the door to real-time decision taking at the nodes based on this knowledge.

\appendices
\section{Distillation on Werner states \label{app:distillation}}
In this appendix, we find the state after successful entanglement distillation on two Werner states.
Performing entanglement distillation on two Werner states with Bell-state fidelities $F_A$ and $F_B$ yields a state with Bell-state fidelity \cite{duer2007entanglement}
\begin{equation}
\label{eq:Fdist}
	\frac{
	\left(F_A F_B + \frac{1}{9}\bar{F}_A \bar{F}_B\right)
	}{
p_{\textnormal{dist}}
}
\end{equation}
where the probability of success $p_{\textnormal{dist}}$ is given by
\begin{equation}
\label{eq:pdistF}
	F_A F_B + \frac{1}{3} F_A \bar{F}_B + \frac{1}{3} \bar{F}_A F_B + \frac{5}{9}\bar{F}_A\bar{F}_B
\end{equation}
where we have denoted $\bar{F} = 1 - F$.
Although the output state is not a Werner state, it is always possible to transform it into a Werner state with the same Bell-state fidelity by local operations.
We rewrite eqs.~\eqref{eq:Fdist} and \eqref{eq:pdistF} as function of the Werner parameters $\werner_A$ and $\werner_B$ of the input states rather than their fidelities $F_A$ and $F_B$ using eq.~\eqref{eq:fidelity-werner-states}, which yields eqs.~\eqref{eq:werner-distillation} and \eqref{eq:pdist}.

\section{Computation of $\Tupper_n$ \label{app:bounds}}

In this appendix, we first prove proposition~\ref{prop:mean-computation} and subsequently show how a modified version of the deterministic algorithm from section~\ref{sec:deterministic} computes the empirical mean $\meanpartial{\Tupper_n}{\ttrunc}$ from eq.~\eqref{eq:computed-mean}.

\subsection{Proof of proposition~\ref{prop:mean-computation} \label{app:dominance}}

The random variable $\Tupper_n$ for $n \in \{0, 1, 2, \dots\}$ is recursively defined as
\begin{eqnarray}
	\Tupper_0 &=& T_0\\
	\label{eq:Tupper0}
	\Mupper_{n+1} &=&
	\left(\Tupper_{n}\right)^{(A)}
	+
	\left(\Tupper_{n}\right)^{(B)}
	\label{eq:Mupper}
	\\
	\Tupper_{n+1} &=& \sum_{k=1}^{K_n} \left(\Mupper_n\right)^{(k)}
	\label{eq:Tupperplus}
\end{eqnarray}
where $T_0$ is defined in section~\ref{sec:Tn} and $K_n$ is geometrically distributed with parameter $\pgen$ for all $n$.

For random variables $X$ and $Y$, both defined on a subset $\mathcal{D}$ of the nonnegative integers, we say that the random variable $Y$ stochastically dominates the random variable $X$, denoted by $X \stoch Y$, if \mbox{$\Pr(X \leq x) \geq \Pr(Y \leq x)$} for all $x \in \mathcal{D}$.
We prove that $\Tupper_n$ stochastically dominates $T_n$ for all $n \geq 0$, for which we need the following lemma.
\begin{lemma}
\label{lemma:dominance-general}
	Let $X, Y, A$ and $B$ each be discrete random variables taking values in the nonnegative integers and let $X'$ ($Y'$) denote an i.i.d. copy of $X$ ($Y$).
Then the following hold:
	\begin{enumerate}[(a)]
		\item \label{lemma:dominance-max-sum} If $X \stoch Y$, then \mbox{$\max\left\{X, X'\right\} \stoch Y + Y'$}.
	\item \label{lemma:axay} If $X \stoch Y$, then $A + X \stoch A + Y$.
	\item \label{lemma:axby} If $X \stoch Y$ and $A \stoch B$, then $A + X \stoch B + Y$.
	\item \label{lemma:dominance-sums} If $m \in \{1, 2, \dots \}$ and $X \stoch Y$, then \mbox{$\sum_{j=1}^m X^{(j)} \stoch \sum_{j=1}^m Y^{(j)}$}.
	\item \label{lemma:dominance-geometric-sum} If $K$ and $K'$ are i.i.d. geometric random variables with parameter $p$ and $X \stoch Y$, then \mbox{$\sum_{j=1}^K X^{(j)} \stoch \sum_{j=1}^{K'} Y^{(j)}$}
\end{enumerate}
where we use the notation $X^{(.)}$ to denote an i.i.d. copy of $X$, following sec.~\ref{sec:notation}.
\end{lemma}
\begin{proof}
For statement \ref{lemma:dominance-max-sum}, we explicitly use that $Y$ only takes nonnegative values so that we can write
	\[
\Pr(Y + Y' \leq y) = \sum_{z=0}^{y} \Pr(Y \leq y - z) \Pr(Y' = z)
.
		\]
	which is, by the fact that any cumulative distribution function is monotone increasing, smaller than
	\begin{eqnarray*}
\sum_{z=0}^{y} \Pr(Y \leq y) \Pr(Y' = z)
		&=&
\Pr(Y \leq y)^2
\\
		&\leq&
\Pr(X \leq y)^2
\\
		&=&
	\Pr(\max\left\{X, X'\right\} \leq y)
	\end{eqnarray*}
	where the inequality is immediate by $X \stoch Y$.
	Statement~\ref{lemma:axay} is proven as
	\begin{eqnarray*}
	\Pr(A + X \leq z)
		&=& \sum_{a=0}^{\infty} \Pr(A=a) \Pr(X \leq z - a)
		\\
		&\geq& \sum_{a=0}^{\infty} \Pr(A=a) \Pr(Y \leq z - a)
		\\
		&=&
		\Pr(A + Y \leq z)
	\end{eqnarray*}
	and statement~\ref{lemma:axby} follows by repeated application of \ref{lemma:axay}:
	\[
		A + X \stoch
		A + Y =
		Y + A 
		\stoch
		Y + B
		= B + Y
		.
		\]
	Statement~\ref{lemma:dominance-sums} can be proven using the fact that \mbox{$\sum_{j=1}^{m} X^{(j)} = X^{(m)} + \sum_{j=1}^{m-1} X^{(j)}$} and statement~\ref{lemma:axby} by induction on $m$.
	For proving statement~\ref{lemma:dominance-geometric-sum}, first note that $\sum_{k=1}^K X^{(k)}$ where $K$ is geometrically distributed with parameter $p$ has cumulative distribution function
\[
	\Pr\left(\sum_{k=1}^K X^{(k)} \leq x\right) = p \cdot \sum_{k=1}^{\infty} (1 - p)^k \cdot \Pr(\sum_{j=1}^k X^{(j)} \leq x)
	\]
	which is a linear combination of the functions \mbox{$f^X_k : x \mapsto \Pr\left(\sum_{j=1}^k X^{(j)} \leq x\right)$}.
	Positivity of the weights \mbox{$p\cdot (1-p)^k$} together with the fact that the \mbox{$f^X_m (x) \geq f^Y_m (x)$} for all $m\in \{1, 2, \dots\}$ and all $x\in \{0, 1, \dots\}$ (see \ref{lemma:dominance-sums}) imply \ref{lemma:dominance-geometric-sum}.
\end{proof}

Using lemma~\ref{lemma:dominance-general}, it is straightforward to prove that $\Tupper_n$ stochastically dominates $T_n$.

\begin{prop}
It holds that $T_n \stoch \Tupper_n$ for all $n \geq 0$.
\end{prop}
\begin{proof}
We use induction on $n$.
	The base case $n=0$ is immediate since $\Tupper_0 = T_0$ (eq.~\eqref{eq:Tupper0}).
Now suppose that $T_{n} \stoch \Tupper_{n}$ for some $n \geq 0$.
	It follows directly from lemma~\ref{lemma:dominance-general}\ref{lemma:dominance-max-sum} that $M_{n} \stoch \Mupper_{n}$, where $M_n$ is given in eq.~\eqref{eq:max} and $\Mupper_n$ in eq.~\eqref{eq:Mupper}.
	Using lemma~\ref{lemma:dominance-general}\ref{lemma:dominance-geometric-sum}, we find that the dominance $M_{n} \stoch \Mupper_{n}$ implies \mbox{$T_{n + 1} \stoch \Tupper_{n + 1}$} where $T_{n+1}$ and $\Tupper_{n+1}$ are defined in eqs.~\eqref{eq:waiting-time} and \eqref{eq:Tupperplus}, respectively.
	This concludes our proof.
\end{proof}

Using this stochastic dominance on the waiting time on each nesting level, we are now ready to prove proposition~\ref{prop:mean-computation}.
First, notice that, in contrast to $T_n$, the mean of $\Tupper_n$ can be computed analytically.

\begin{lemma}
\label{lemma:mean-Tupper}
	\[
	\mean{\Tupper_n} = \left(\frac{2}{\pswap}\right)^n \cdot \frac{1}{\pgen}
	\]
\end{lemma}
\begin{proof}
We use induction on $n$.
	Since $\Tupper_0$ equals $T_0$, which is geometrically distributed with parameter $\pgen$, we have $\mean{\Tupper_0} = 1 / \pgen$.
	For the induction step, first note that by linearity of the mean, we have
	\[
		\mean{\Mupper_{n}} = \mean{\left(\Tupper_n\right)^{(A)}} + \mean{\left(\Tupper_n\right)^{(B)}} = 2\mean{\Tupper_{n}}.
	\]
	The last step is given by Wald's identity \cite{wald1947sequential}, which states that the mean of a compound sum equals the product of the mean of the summand and the random variable that is the summation upper bound:
	\[
		\mean{\Tupper_{n + 1}} = \mean{K_{n}} \cdot \mean{\Mupper_{n}} = \frac{1}{\pswap} \cdot 2\mean{\Tupper_{n}}.
		\]
	Closing the recursion relation on $\mean{\Tupper_{n}}$ yields the expression in the lemma.
\end{proof}

The following lemma, which states that stochastic domination implies domination of the empirical mean from eq.~\eqref{eq:computed-mean} provides the last step in proving proposition~\ref{prop:mean-computation}.

\begin{lemma}
\label{lemma:meanbound}
	Let $X$ and $Y$ be discrete random variables, both defined on the nonnegative integers.
	If $X \stoch Y$, then
	\[
		0 \leq \mean{X} - \meanpartial{X}{\ttrunc} \leq \mean{Y} - \meanpartial{Y}{\ttrunc}
		\]
	for each $\ttrunc \in \{0, 1, \dots\}$.
	In particular, for $\ttrunc = 0$ it follows that \mbox{$\mean{X} \leq \mean{Y}$}.
\end{lemma}
\begin{proof}
	The lower bound is an immediate consequence of positivity of probabilities and
	\begin{equation}
		\label{eq:difference}
		\mean{X} - \meanpartial{X}{\ttrunc}=	\sum_{t = \ttrunc + 1}^{\infty} \Pr(X \geq t)
		.
	\end{equation}
	The upper bound follows from eq.~\eqref{eq:difference} and the definition of stochastic dominance: \mbox{$\Pr(X \geq t) \leq \Pr(Y \geq t)$} for all \mbox{$t\in \{0,1,\dots\}$}.
\end{proof}

Proposition~\ref{prop:mean-computation} follows from lemma~\ref{lemma:meanbound} by replacing $X$ by $T_n$ and $Y$ by $\Tupper_n$, and substituting $\mean{\Tupper_n}$ by the expression in lemma~\ref{lemma:mean-Tupper}.

\subsection{Computing the empirical mean of $\Tupper_n$ \label{app:compute-Tuppern}}
Here, we outline how the deterministic algorithm computes $\meanpartial{\Tupper_n}{\ttrunc}$, which is needed for determining a bound on the mean of $T_n$ using proposition~\ref{prop:mean-computation}.
First note that $\Tupper_n$ and $T_n$ are identical except for the difference between $M_n$ in eq.~\eqref{eq:max}, which equals the maximum of two copies of the waiting time, and $\Mupper_n$ in eq.~\eqref{eq:Mupper}, which equals their sum.
We modify the algorithm to account for this difference by replacing the computation of the probability distribution of $M_n$ in eq.~\eqref{eq:max-step-computation} by the convolution
\[
	\Pr(\Mupper_n = t) = \sum_{j=0}^t \Pr(\Tupper_{n-1} = j) \Pr(\Tupper_{n-1} = t - j)
	.
	\]
In order to determine $\meanpartial{\Tupper_n}{\ttrunc}$, we first run the modified algorithm to compute the cumulative probability distribution $\Pr(\Tupper_n \leq t)$ for $t\in \{0, 1, \dots, \ttrunc\}$, after which we calculate
\[
	\meanpartial{\Tupper_n}{\ttrunc} = \sum_{t=1}^{\ttrunc} \left[ 1 - \Pr(\Tupper_n \leq t - 1)\right]
	.
	\]

\section*{Acknowledgment}
The authors would like to thank Kenneth Goodenough, Alfons Laarman, Filip Rozp\k{e}dek, Joshua Slater and Stephanie Wehner for helpful discussions.
This work was supported by the QIA project (funded by European Union's Horizon 2020, Grant Agreement No. 820445) and by the Netherlands Organization for Scientific Research (NWO/OCW), as part of the Quantum Software Consortium program (project number 024.003.037 / 3368).

\ifCLASSOPTIONcaptionsoff
  \newpage
\fi

\bibliographystyle{IEEEtran}
\bibliography{allbibs}
\end{document}